\documentclass[a4paper]{article}

\usepackage{amsmath,amsfonts,amssymb,enumitem,amsthm,color,hyperref,anysize,graphicx,epstopdf,algorithm,algpseudocode,cleveref,multirow,url,extarrows}
\usepackage[usenames,dvipsnames]{xcolor}
\usepackage[margin=0.81in]{geometry}

\newcommand{\rr}{\mathbb{R}}
\newcommand{\rrplus}{\rr^+}

\newcommand{\ep}{\epsilon}

\newcommand{\ra}{\rightarrow}
\newcommand{\floor}[1]{\left\lfloor #1 \right\rfloor}
\newcommand{\ceil}[1]{\left\lceil #1 \right\rceil}

\newcommand{\expect}[1]{\mathbb{E}\left[ #1 \right]}

\renewcommand{\bold}[1]{\textbf{#1}}
\newcommand{\parabold}[1]{\noindent\bold{#1}}

\newcommand{\hide}[1]{}

\newcommand{\interior}{\mathsf{int}}
\newcommand{\bbone}{\mathbf{1}}

\newcommand{\tG}{\mathsf{G}}
\newcommand{\tGinv}{\mathsf{G}^{-1}}
\newcommand{\hC}{\overline{C}}

\DeclareMathOperator*{\argmin}{arg\,min}
\DeclareMathOperator*{\argmax}{arg\,max}

\newcommand{\difft}[1]{\frac{\mathsf{d} #1}{\mathsf{d} t}}

\newenvironment{pf}{\begin{proof}[\emph{\textbf{Proof: }}]}{\end{proof}}
\newenvironment{pfof}[1]{\begin{proof}[\emph{\textbf{Proof of #1: }}]}{\end{proof}}

\newtheorem{theorem}{Theorem}
\newtheorem{lemma}[theorem]{Lemma}
\newtheorem{proposition}[theorem]{Proposition}
\newtheorem{definition}[theorem]{Definition}

\newcommand{\calC}{\mathcal{C}}
\newcommand{\calD}{\mathcal{D}}
\newcommand{\calE}{\mathcal{E}}

\newcommand{\calO}{\mathcal{O}}

\newcommand{\bbp}{\mathbf{p}}
\newcommand{\bbq}{\mathbf{q}}

\newcommand{\bbs}{\mathbf{s}}

\newcommand{\bbu}{\mathbf{u}}
\newcommand{\bbv}{\mathbf{v}}

\newcommand{\bbx}{\mathbf{x}}
\newcommand{\bby}{\mathbf{y}}
\newcommand{\bbz}{\mathbf{z}}

\newcommand{\bbA}{\mathbf{A}}
\newcommand{\bbB}{\mathbf{B}}

\newcommand{\bbI}{\mathbf{I}}
\newcommand{\bbJ}{\mathbf{J}}

\newcommand{\bbM}{\mathbf{M}}

\newcommand{\bbR}{\mathbf{R}}

\newcommand{\trans}{^{\mathsf{T}}}

\newcommand{\hbbx}{\hat{\bbx}}

\newcommand{\hbby}{\hat{\bby}}

\newcommand{\calEd}{\calE^{\delta}}

\newcommand{\vol}{\mathsf{vol}}
\newcommand{\hj}{\hat{j}}
\newcommand{\hk}{\hat{k}}

\title{Chaos, Extremism and Optimism:\\Volume Analysis of Learning in Games}

\author{
Yun Kuen Cheung\\
Singapore University of\\
Technology and Design
\and
Georgios Piliouras\\
Singapore University of\\
Technology and Design
}
\date{}

\begin{document}

\maketitle

\begin{abstract}
We present volume analyses of Multiplicative Weights Updates (MWU) and Optimistic Multiplicative Weights Updates (OMWU) in zero-sum as well as coordination games.
Such analyses provide new insights into these game dynamical systems, which seem hard to achieve via the classical techniques within Computer Science and Machine Learning.

The first step is to examine these dynamics not in their original space (simplex of actions) but in a dual space (aggregate payoff space of actions).
The second step is to explore how the volume of a set of initial conditions evolves over time when it is pushed forward according to the algorithm.
This is reminiscent of approaches in Evolutionary Game Theory where replicator dynamics, the continuous-time analogue of MWU, is known to always preserve volume in all games.
Interestingly, when we examine discrete-time dynamics, both the choice of the game and the choice of the algorithm play a critical role.
So whereas MWU expands volume in zero-sum games and is thus Lyapunov chaotic, we show that OMWU contracts volume, providing an alternative understanding for its known convergent behavior.
However, we also prove a no-free-lunch type of theorem, in the sense that when examining coordination games the roles are reversed:
OMWU expands volume exponentially fast, whereas MWU contracts.

Using these tools, we prove two novel, rather negative properties of MWU in zero-sum games:
\begin{enumerate}[leftmargin=0.2in]
\item Extremism: even in games with unique fully mixed Nash equilibrium, the system recurrently gets stuck near pure-strategy profiles,
despite them being clearly unstable from game theoretic perspective.
\item Unavoidability: given any set of \emph{good} points (with your own interpretation of ``good''), the system cannot avoid \emph{bad} points indefinitely.
\end{enumerate}
\end{abstract}

\begin{figure}[htp]
\begin{center}
\includegraphics[scale=0.4]{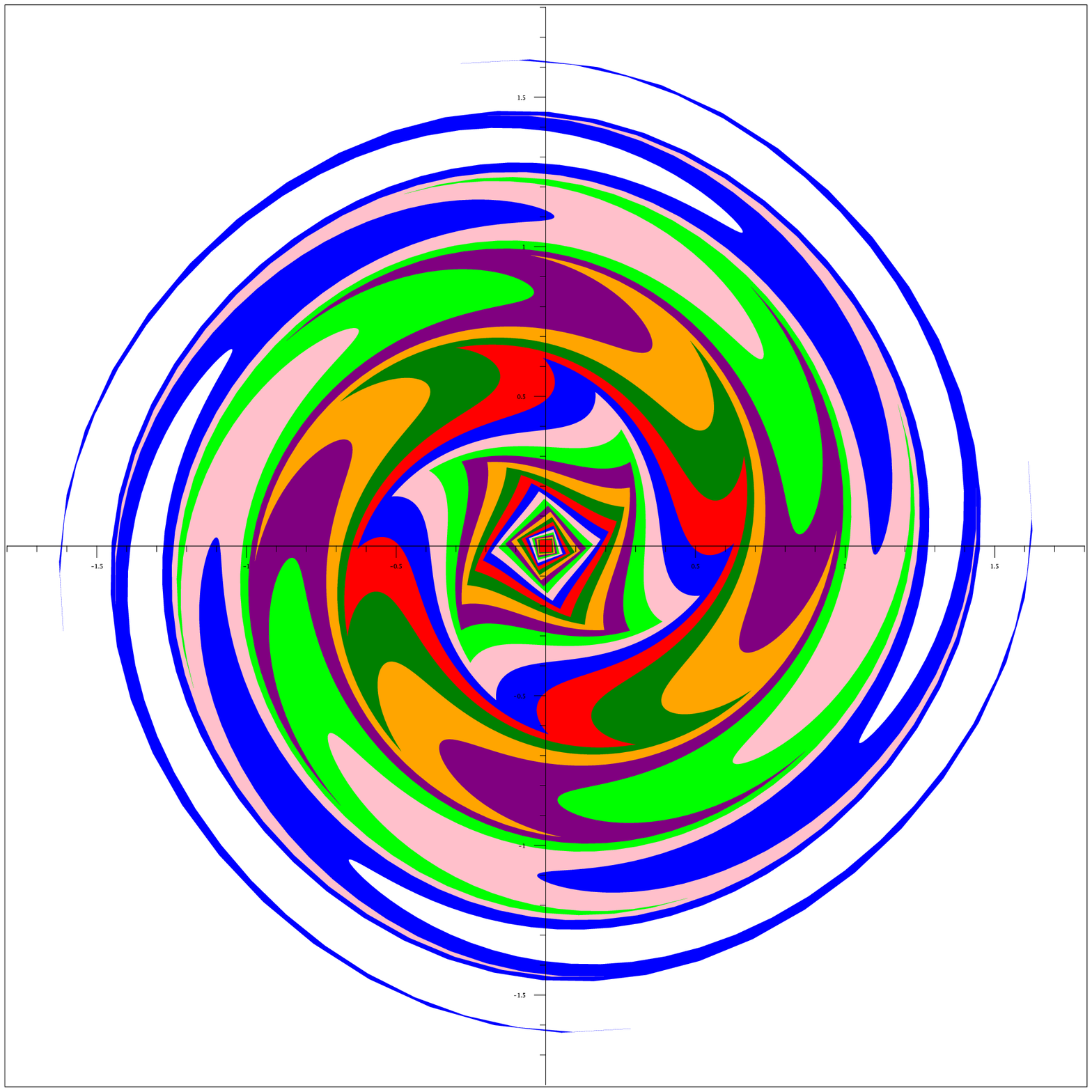}~
\includegraphics[scale=0.4]{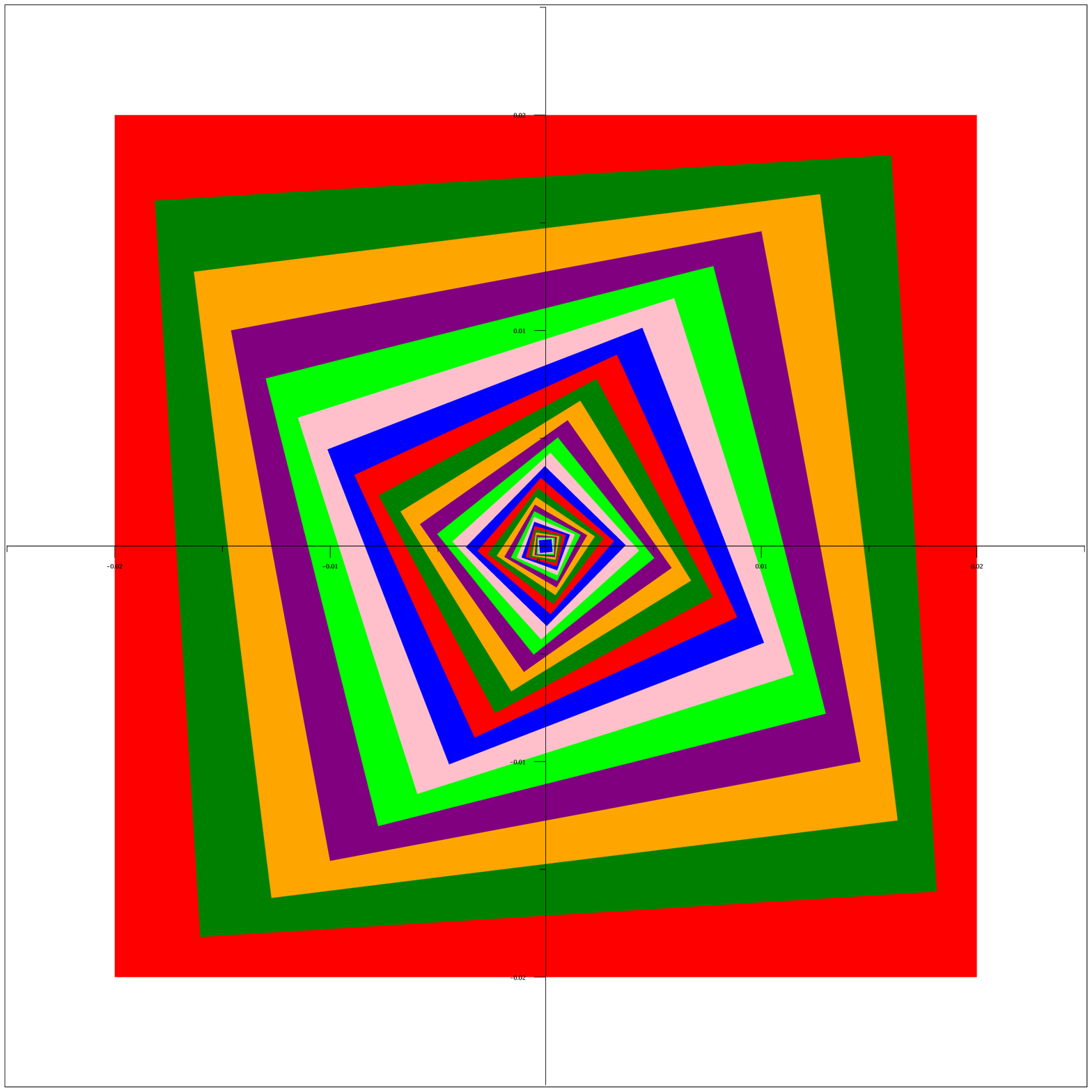}\\
\end{center}
\caption{Evolution of MWU (left) and OMWU (right) in Matching-Pennies game in the dual space. The origin is the unique Nash equilibrium. In this 2-D system, volume is area.
\underline{MWU:} The initial set is the tiny red square around the equilibrium. When the set is evolved via MWU, it rotates  tornado-like and its area increases.~\cite{CP2019}
\underline{OMWU:} The initial set is the outermost red square. When the set is evolved via OMWU, it shrinks toward the equilibrium and its area decreases.}\label{fig:tornado}
\end{figure}

\section{Introduction}\label{sect:intro}

In recent years, fuelled by AI applications such as Generative Adversarial Networks (GANs),
there has a been a strong push towards a more detailed understanding of the behavior of online learning dynamics in zero-sum games and beyond.
Even when focusing on the canonical case of bilinear zero-sum games, the emergent behavior depends critically on the choice of the training algorithms.
Results can macroscopically be grouped in three distinct categories: convergent, divergent and cyclic/recurrent. 
Specifically, most standard regret minimizing dynamics and online optimization dynamics, such as Multiplicative Weights Updates (MWU) or gradient descent~\cite{Cesa06},
although their time average converges~\cite{freund1999adaptive}, their day-to-day behavior diverges away from Nash equilibria~\cite{BP2018,Cheung2018}.
On the other hand, some game-theoretically inspired dynamics, such as Optimistic Multiplicative Weights Updates (OMWU) converge~\cite{DaskalakisP2019,daskalakis2018training}.
(Numerous other convergent heuristics have also been recently analyzed, e.g.~\cite{mertikopoulos2019optimistic,gidel2019a,gidel2019negative,balduzzi2018mechanics,2019arXiv190602027A}.)
Finally, if we simplify learning into continuous-time ordinary differential equations (ODEs), such as Replicator Dynamics, the continuous time analogue of MWU,
the emergent behavior becomes almost periodic (Poincar\'{e} recurrence)~\cite{PS2014,GeorgiosSODA18,boone2019darwin}.
%Informally speaking, every initial condition returns infinitely often arbitrarily close to itself. 
This level of complex case-by-case analysis just to understand bilinear zero-sum games seems daunting.
Can we find a more principled approach behind these results that is also applicable to more general games?

One candidate is \emph{volume analysis}, a commonly used tool in the area of Dynamical Systems.
Briefly speaking, what it does is to consider a set of starting points with positive volume (Lebesgue measure),
and analyze how the volume changes as the set evolves forward in time. %according to a numerical algorithm.
As we shall see, an advantage of volume analysis is its general applicability, for it can be used to analyze not just ODEs
but different discrete-time algorithms such as MWU and OMWU in different types of games.

In Evolutionary Game Theory, volume analysis has been applied to \emph{continuous-time} dynamical systems/ODEs (see~\cite[Sections 11.3 and 13.5]{HS1998} and~\cite[Chapter 9]{Sandholm10}).
Eshel and Akin~\cite{EA1983} showed that RD in any multi-player matrix game is volume preserving in the dual (aggregate payoff) space.
This result is in fact a critical step in the proof of Poincar\'{e} recurrence in zero-sum games.
Loosely speaking, if we think of the set of initial conditions as our uncertainty about where is the starting point of the system,
since uncertainty does not decrease (convergence) or increase (divergence) we end up cycling in space.
 
Cheung and Piliouras~\cite{CP2019} recently applied volume analysis to \emph{discrete-time} numerical algorithms in a series of games,
including two-person zero-sum games, graphical constant-sum games, generalized Rock-Paper-Scissors games and general $2\times 2$ bimatrix games.
Among other results, they showed that MWU in zero-sum games is \emph{Lyapunov chaotic} in the dual space.
This is done by showing that the volume of any set is expanding exponentially fast. 
Lyapunov chaos is one of the most classical notions in the area of Dynamical Systems that captures instability and unpredictability.
More precisely, it captures the following type of \emph{butterfly effect}: when the starting point of a dynamical system is slightly perturbed,
the resulting trajectories and final outcomes diverge quickly.
Lyapunov chaos means that such system is very sensitive to round-off errors in computer simulations; thus, the result of Cheung and Piliouras
provides a rigorous mathematical explanation to the numerical instability of MWU experiments.

\medskip

\parabold{Our Contributions, and Roadmap for This Paper.}
Our contributions can be summarized into two categories, both stemming from volume analyses.
First, besides the numerical instability and unpredictability already mentioned,
we discover two novel and negative properties of MWU in zero-sum games in this paper,
which are consequences of exponential volume expansion. We call them \emph{unavoidability} and \emph{extremism}.
We have given informal descriptions of these two properties in the abstract; we will give more details about them below.

Second, we carry out volume analysis on OMWU and discover that its volume-change behavior is in stark contrast with MWU.
To understand why we should be interested in such an analysis, we first point out that in the study of game dynamics,
a primary target is to seek algorithms that behave well in as broad family of games as possible.
Recently, OMWU was shown to achieve stability in zero-sum games, despite its strong similarity with MWU (which is chaotic in zero-sum games).
It is natural to ask how these stability behaviors generalize to other games.
We provide a negative answer, by proving that OMWU is volume-expanding and Lyapunov chaotic in coordination games; see Figure~\ref{fig:main} for a summary of this no-free-lunch phenomena.
We show that the volume is exponentially decreasing for OMWU in zero-sum game, mirroring with
the recent stability results~\cite{daskalakis2018training,DaskalakisP2019} in the original (primal) space (simplex of actions) about these game dynamics.
The details are presented in Section~\ref{sect:volume-contract}.

As unavoidability and extremism are shown via volume expansion argument, it is easy to generalize and show that these two properties also appear in OMWU in coordination games.
On a technical note, the volume analysis on OMWU is more involved than that on MWU. Along the process, we propose an ODE system which is the continuous-time analogue of OMWU in games,
and the volume analysis relies crucially on the fact that discrete-time OMWU is an \emph{online Euler discretization} of the ODE system; see Section~\ref{sect:cont-analog} for details.

\begin{figure}[htp]
\begin{center}
\begin{tabular}{ |c|c|c| } 
\hline
& Zero-sum Games & Coordination Games \\ 
\hline
MWU & $\mathbf{+}$~\cite{CP2019} & $-$ (Theorem~\ref{thm:MWU-volume-contract-coordination-discrete}) \\ 
OMWU & $-$ (Theorem~\ref{thm:OMD-volume-contract-zerosum-discrete}) & $\mathbf{+}$ (Theorem~\ref{thm:OMD-volume-expand-coordination-discrete}) \\ 
\hline
\end{tabular}
\end{center}
\caption{How volume changes in the dual space. ``$+$'' denotes exponential volume expansion, unavoidability and extremism, while ``$-$'' denotes exponential volume contraction.
See Figures~\ref{fig:tornado} and~\ref{fig:coor} for graphical illuminations.}
\label{fig:main}
\end{figure}

In Section~\ref{sect:volume-analysis-prelim}, we discuss how volume analyses can be carried out on learning algorithms that are \emph{gradual}, i.e.~controlled by a step-size $\ep$.
We demonstrate that volume analyses can often be boiled down to analyzing a polynomial of $\ep$ that arises from the expansion of the determinant of some \emph{Jacobian matrix}.
This convincingly indicates that volume analyses can be readily applicable to a broad family of learning algorithms.

In the rest of this introduction, we discuss extremism and unavoidability with more details.
Section~\ref{sect:prelim} contains the necessary background for this work. All missing proofs can be found in the appendix.

\medskip

\parabold{Extremism (Section~\ref{sect:extremism}).}
The more formal statement of our extremism theorem (Theorem~\ref{thm:extremism}) is: given any zero-sum game that satisfies mild regularity condition,
there is a dense set of starting points from which MWU will lead to a state where both players concentrate their game-plays on only one strategy.
More precisely, let $\bbx,\bby$ denote the mixed strategies of the two players.
For any $\delta > 0$, there is a dense set of starting points $(\bbx^0,\bby^0)$,
from which MWU with a suitably small step-size leads to $(\bbx^t,\bby^t)$ for some time $t$,
where there exists a strategy $j$ of Player 1 with $x^t_j \ge 1-\delta$, and a strategy $k$ of Player 2 with $y^t_k \ge 1-\delta$.

To understand how bizarre extremism is, consider the classical Rock-Paper-Scissors game,
which is a zero-sum game with a unique fully-mixed Nash equilibrium, where each strategy is chosen with equal probability $1/3$.
The extremism theorem indicates that there exists a starting point arbitrarily close to the Nash equilibrium,
which will eventually lead to a situation where each player essentially sticks with one strategy for a long period
of time\footnote{When $x_j^t < \delta$, it takes at least $\Omega \left( \frac 1\ep \ln \frac{1}{\delta} \right)$ time before $x_j$ can possibly resume a ``normal'' value,
say above $1/20$.}.
As no pure Nash equilibrium exists, the trajectory will recurrently approach and then escape such extremal points infinitely often (Theorem~\ref{thm:extremal-io}),
demonstrating that the dynamic is very unstable.

\medskip

\parabold{Unavoidability (Section~\ref{sect:unavoidability}).}
The extremism theorem is actually an indirect consequence of an unavoidability theorem of MWU in zero-sum games.
Unavoidability is a notion first introduced in a topic of (automatic) control theory called ``avoidance control''~\cite{LS1977}, which addresses the following type of problems:
for dynamical/automatic systems, analyze whether they can always avoid reaching certain \emph{bad} states, e.g.~collisions of robots/cars, or places with severe weather conditions.

To explain unavoidability of MWU in general games, we need another notion of \emph{uncontrollability}.
Let $U$ be a region which is in the strict interior of the primal simplex, and let $V$ be the correspondence set of $U$ in the dual space.
Informally, we say a region $U$ is \emph{uncontrollable} if any subset of $V$ is exponentially volume expanding in the dual space.
As the volume is expanding quickly, it is impossible for $V$ to contain the evolved set after a sufficiently long period of time,
which, when converting back to the primal space, implies that the dynamic escapes from $U$ (Theorem~\ref{thm:unavoidable}).
When $U$ is thought as a set of \emph{good} points and its complements are the set of \emph{bad} points, the above discussions can be summarized by a punchline:
\begin{center}
\emph{When a good set is uncontrollable, the bad set is unavoidable.}
\end{center}

Note that the above discussion concerns general games. When we narrow down to zero-sum games, the results of Cheung and Piliouras~\cite{CP2019} indicate that
under mild regularity condition, \emph{any} set $U$ in the the strict interior of the primal simplex is uncontrollable.
Thus, for MWU in zero-sum game, you can have whatsoever interpretation of ``good'', but the corresponding bad set is unavoidable.

Some ideas behind the proof of unavoidability come from Cheung and Piliouras~\cite{CP2019},
who demonstrated several negative properties of MWU in special games, including generalized Rock-Paper-Scissors games.
Our key contribution here is to formulate and prove the fully generalized statement about this property.
In the proof of the extremism theorem, the first step uses unavoidability to show that extremism appears for one player.
But to show that extremism appears for both players simultaneously, we need some substantially novel arguments in the subsequent steps.
\section{Preliminary}\label{sect:prelim}

\subsection{Games}
In this paper, we focus on two-person general normal-form games. The strategy set of Player $i$ is $S_i$. Let $n=|S_1|$ and $m=|S_2|$. We assume $n,m\ge 2$ throughout this paper.
Let $\bbA$ and $\bbB$ be two $S_1\times S_2$ matrices, which represent the payoffs to Players 1 and 2 respectively.
We assume all payoffs are bounded within the interval $[-1,1]$.
Let 
$\Delta^d := \left\{(z_1,z_2,\cdots,z_d)\in \rr^d ~\big|~ \sum_{j=1}^d z_j = 1,~~\text{and}~~\forall~j,~z_j\ge 0\right\}$.
We call $\Delta := \Delta^n \times \Delta^m$ the \emph{primal simplex} or \emph{primal space} of the game, which contains the set of all mixed strategy profiles of the two players.
We use $\bbx\in \Delta^n$ and $\bby\in \Delta^m$ to denote strategy vectors of Players 1 and 2 respectively.
When a zero-sum game is concerned, only the matrix $\bbA$ needs to be specified, as $\bbB = -\bbA$.

\begin{definition}\label{def:trivial}
A zero-sum game $(\bbA,-\bbA)$ is \emph{trivial} if there exist real numbers $a_1,a_2,\cdots,a_n$ and $b_1,b_2,\cdots,b_m$ such that $A_{jk} = a_j + b_k$.
\end{definition}

A trivial game is not interesting as each player has a clear dominant strategy; for Player 1 it is $\argmax_{j\in S_1} a_j$, while for Player 2 it is $\argmin_{k\in S_2} b_k$.
Following~\cite{CP2019}, we measure the distance of a zero-sum game $\bbA$ from triviality by
\begin{equation}\label{eq:distance-from-trivial}
c(\bbA) ~:=~ \min_{\substack{a_1,\cdots,a_n\in \rr\\b_1,\cdots,b_m\in \rr}}~~\left[\max_{\substack{j\in S_1\\k\in S_2}}~ \left(A_{jk} - a_j - b_k\right)
~~-~~ \min_{\substack{j\in S_1\\k\in S_2}}~ \left(A_{jk} - a_j - b_k\right)\right].
\end{equation}
Observe that if $\bbA'$ is a sub-matrix of $\bbA$, then $c(\bbA') \le c(\bbA)$. If one of the two dimensions of $\bbA$ is one, then $c(\bbA) = 0$.
By setting all $a_j,b_k$ to zero, we have the trivial bound $c(\bbA) \le 2$.

For a coordination game, i.e.~a game with payoff matrices in the form of $(\bbA,\bbA)$, we also measure its distance from triviality using Equation~\eqref{eq:distance-from-trivial}.

\subsection{MWU and OMWU Update Rules in Dual and Primal Spaces}

As is well-known, MWU and OMWU can be implemented either in the primal space, or in a dual space.
The dual space is $\calD := \rr^n\times \rr^m$, in which MWU with positive step size $\ep$ generates a sequence of updates $(\bbp^0,\bbq^0),(\bbp^1,\bbq^1),$
$(\bbp^2,\bbq^2),\cdots$,
where $p_j^t - p_j^0$ is $\ep$ times the cumulative payoff to Player 1's strategy $j$ up to time $t$,
and $q_k^t - q_k^0$ is $\ep$ times the cumulative payoff to Player 2's strategy $k$ up to time $t$.
At each time step, a point $(\bbp^t,\bbq^t)\in \calD$ is converted to a point $(\bbx^t,\bby^t)=(\bbx(\bbp^t),\bby(\bbq^t))\in \Delta$ by the following rules:
\begin{equation}\label{eq:dual-to-primal}
x_j^t = x_j(\bbp^t) = \exp(p_j^t)\left/\left(\sum_{\ell\in S_1} \exp(p_\ell^t)\right)\right.~~~\text{and}~~~y_k^t = y_k(\bbq^t) = \exp(q_k^t)\left/\left(\sum_{\ell\in S_2} \exp(q_\ell^t)\right)\right..
\end{equation}
For convenience, we let $\tG$ denote the function that converts a dual point to a primal point, i.e.~$\tG(\bbp,\bbq) = (\bbx(\bbp),\bby(\bbq))$.

For MWU in a general-sum game, the payoffs to Player 1's all strategies in round $(t+1)$ can then be represented by the vector $\bbA \cdot \bby(\bbq^t)$,
while the payoffs to Player 2's all strategies in round $(t+1)$ can be represented by the vector $\bbB\trans \cdot \bbx(\bbp^t)$.
Accordingly, the MWU update rule in the dual space can be written as
\begin{equation}\label{eq:MWU}
\bbp^{t+1} ~=~ \bbp^t + \ep\cdot \bbA \cdot \bby(\bbq^t)~~~~~~\text{and}~~~~~~
\bbq^{t+1} ~=~ \bbq^t + \ep\cdot \bbB\trans \cdot \bbx(\bbp^t).
\end{equation}
The above update rule in the dual space is equivalent to the following MWU update rule in the primal space with starting point $\tG(\bbp^0,\bbq^0)$,
which some readers might be more familiar with:
\begin{equation}\label{eq:MWU-primal}
x_j^{t+1} ~=~ \frac{x_j^t \cdot \exp(\ep \cdot [\bbA \cdot \bby^t]_j)}{\sum_{\ell\in S_1} x_\ell^t \cdot \exp(\ep \cdot [\bbA \cdot \bby^t]_\ell)}~~~~~~\text{and}~~~~~~
y_k^{t+1} ~=~ \frac{y_k^t \cdot \exp(\ep \cdot [\bbB\trans \cdot \bbx^t]_k)}{\sum_{\ell\in S_2} y_\ell^t \cdot \exp(\ep \cdot [\bbB\trans \cdot \bbx^t]_\ell)}.
\end{equation}
%We point out a simple fact which will be useful later: $\exp(-2\ep) \le x_j^{t+1}/x_j^t~,~y_k^{t+1}/y_k^t \le \exp(2\ep)$.

For OMWU in a general-sum game with step-size $\ep$,
its update rule in the dual space starts with $(\bbp^0,\bbq^0)=(\bbp^1,\bbq^1)$, % = (\bbp^0 + \bbA \cdot \bby(\bbq^0),\bbq^0 + \bbB\trans \cdot \bbx(\bbp^0))
and for $t\ge 1$,
\begin{equation}\label{eq:OptMD}
\bbp^{t+1} = \bbp^t + \ep\cdot \left[2\bbA \cdot \bby(\bbq^t)-\bbA\cdot \bby(\bbq^{t-1})\right]~~\text{and}~~
\bbq^{t+1} = \bbq^t + \ep\cdot \left[2\bbB\trans \cdot \bbx(\bbp^t) - \bbB\trans \cdot \bbx(\bbp^{t-1})\right],
\end{equation}
where $\bbx(\bbp^t),\bby(\bbq^t)$ are as defined in~\eqref{eq:dual-to-primal}.
The above update rule in the dual space has an equivalent update rule in the primal space, which is the same as~\eqref{eq:MWU-primal},
except we replace $\bbA \cdot \bby^t$ there by $2\bbA \cdot \bby^t-\bbA\cdot \bby^{t-1}$,
and replace $\bbB\trans \cdot \bbx^t$ there by $2\bbB\trans \cdot \bbx^t - \bbB\trans \cdot \bbx^{t-1}$.

Note that for the update rule~\eqref{eq:OptMD}, for $t\ge 2$, we have
$\bbp^t-\bbp^0 = \ep (\sum_{\tau=1}^{t-2} \bbA \cdot \bby(\bbq^\tau) ~+~ 2\cdot \bbA \cdot \bby(\bbq^{t-1}))$,
which can be viewed as $\ep$ times the cumulative payoff to strategy $j$ from time $2$ up to time $t$,
but with a double weight on the last-iterate payoff.

\subsection{Relationships between Primal and Dual Spaces}\label{subsect:relation}

Here, we clarify some facts about primal and dual spaces and their relationships. 
%While some of these facts are not directly relevant to the results in our papers,
%we feel such clarifications necessary to clear any potential confusions of the readers.
Equation~\eqref{eq:dual-to-primal} provides a conversion from a point in $\calD$ to a point in the interior of the primal space, i.e., $\interior(\Delta)$.
It is not hard to see that there exist multiple points in $\calD$ which convert to the same point in $\interior(\Delta)$.
Precisely, by~\cite[Proposition 1]{CP2019}, if $(\bbp,\bbq),(\bbp',\bbq')\in \calD$, then
$(\bbx(\bbp),\bby(\bbq)) = (\bbx(\bbp'),\bby(\bbq'))$ if and only if
$\bbp-\bbp' = c_1 \cdot \bbone$ and $\bbq-\bbq' = c_2 \cdot \bbone$ for some $c_1,c_2\in \rr$.
%as stated in the proposition below.
For any $S\subset \interior(\Delta)$, we let $\tGinv(S)$ denote the set of points $(\bbp,\bbq)$ in the dual space $\calD$ such that $\tG(\bbp,\bbq)\in S$.

Since the primal and dual spaces are not in one-one correspondence, some readers might argue that the \emph{reduced} dual space used by Eshel and Akin~\cite{EA1983}
(in which its $(n+m-2)$ dual variables denote the quantities $p_1-p_n,p_2-p_n,\cdots,p_{n-1}-p_n,q_1-q_m,q_2-q_m,\cdots,q_{m-1}-q_m$) is a better choice.
Our reason for choosing $\calD$ as the dual space to work with is simply because we are unable to establish the same type of results (like Lemma~\ref{lem:Cxy} below)
for the reduced dual space.

The following proposition shows that volume expansion in the dual space implies large diameter in the primal space,
if the corresponding primal set is bounded away from the simplex boundary.

\begin{proposition}\label{prop:dual-expand-to-primal-instability-simpler}
Let $S$ be a set in the dual space with Lebesgue volume $v$.
Suppose there exists $j\in S_1$ and $k\in S_2$ such that $\max_{(\bbp,\bbq)\in S} p_j - \min_{(\bbp,\bbq)\in S} p_j \le R_j$
and $\max_{(\bbp,\bbq)\in S} q_k - \min_{(\bbp,\bbq)\in S} q_k \le R_k$.
Also, suppose that for some $\kappa > 0$, there exists a point $(\bbx,\bby)\in \tG(S)$ such that either every entry of $\bbx$ is at least $\kappa$
or every entry of $\bby$ is at least $\kappa$.
Then the diameter of $\tG(S)$ %w.r.t.~$\ell_2$ norm 
is at least
$
\left[ 1 - \exp\left(-\frac 14 \cdot \left( \frac{v}{R_j R_k} \right)^{1/(n+m-2)}\right) \right] \cdot \kappa.
$
\end{proposition}

We point out that while we use volume as the mean for analyses, when measuring instability what we really care is the diameter of the set $S$ or its corresponding primal set.
Indeed, volume is not an ideal benchmark, as we present concrete examples in Appendix~\ref{app:missing-prelim} to show that
(A) volume contraction in the dual space does \emph{not} necessarily imply stability in either the dual or the primal space;
(B) volume expansion in the dual space does \emph{not} necessarily imply instability in the primal space if the primal set is near the simplex boundary.
We show that (B) remains true in the reduced dual space.

\subsection{Dynamical System, Jacobian, and Volume of Flow}\label{sect:volume-analysis-prelim}

A dynamical system is typically described by a system of ordinary differential equations (ODE) over time in $\rr^d$,
governed by $d$ differential equations on the variables $s_1,s_2,\cdots,s_d$,
which are of the form  $\forall j\in [d],~~\difft{s_j} = F_j(s_1,s_2,\cdots,s_d).$
Given a starting point $(s_1^\circ,s_2^\circ,\cdots,s_d^\circ)$,
the values of the variables at any time $t\ge 0$ are typically uniquely determined;
precisely, given the starting point, for each $j\in [d]$, there is a function $s_j:\rrplus \ra \rr$
such that altogether they satisfy the ODE system, with $(s_1(0),s_2(0),\cdots,s_d(0))$ being the given starting point.
The collection of the functions $s_1,s_2,\cdots,s_d$ is called the \textit{trajectory} of the given starting point.
The \textit{flow} of a given starting point at time $t$ is simply $(s_1(t),s_2(t),\cdots,s_d(t))$.
In this paper, we assume that $F_j$ is smooth everywhere.

Given a measurable set $S$ and an ODE system, the \textit{flow} of $S$ at time $t$
is simply the collection of the flows of all starting points in $S$ at time $t$;
when the underlying ODE system is clear from context, we denote it by $S(t)$.
Let $\vol(S)$ denote the Lebesgue volume of set $S$.

The \textit{Jacobian} of the ODE system at $\bbs=(s_1,s_2,\cdots,s_d)$ is the $d\times d$-matrix $\bbJ(\bbs)$:
\begin{equation}\label{eq:Jacobian}
\bbJ(\bbs) ~=~ 
\begin{bmatrix}
\frac{\partial}{\partial s_1}F_1(\bbs) & \frac{\partial}{\partial s_2}F_1(\bbs) & \cdots & \frac{\partial}{\partial s_d}F_1(\bbs)\\
%\frac{\partial}{\partial s_1}F_2(\bbs) & \frac{\partial}{\partial s_2}F_2(\bbs) & \cdots & \frac{\partial}{\partial s_d}F_2(\bbs)\\
\vdots & \vdots & \ddots & \vdots \\
\frac{\partial}{\partial s_1}F_d(\bbs) & \frac{\partial}{\partial s_2}F_d(\bbs) & \cdots & \frac{\partial}{\partial s_d}F_d(\bbs)
\end{bmatrix}.
\end{equation}
The Liouville's theorem states that if $S(0)\subset \rr^d$ is a bounded and measurable set,
then
\begin{equation}\label{eq:Liouville-cont}
\difft{}\vol(S(t)) = \int_{\bbs\in S} \mathsf{trace}(\bbJ(\bbs))\,\mathsf{d}V.
\end{equation}

The Liouville's theorem is indeed the continuous analogue of integration by substitution for multi-variables,
which applies for calculating volume changes of discrete-time update rules.
We present a simplified version of it which suffices for our purposes.
For a \emph{gradual} update rule
\[
\bbs_{t+1} = \bbs_t + \ep \cdot F(\bbs_t),
\]
where $F:\rr^d\ra \rr^d$ is a smooth function and step-size $\ep>0$,
if $S\subset \rr^d$ is a bounded and measurable set, and if the discrete flow in one time step maps $S$ to $S'$ injectively, then
\begin{equation}\label{eq:Liouville-discrete}
\vol(S') ~=~ \int_{\bbs\in S} \det \left( \bbI + \ep \cdot \bbJ(\bbs) \right) \,\mathsf{d}V,
\end{equation}
where $\bbJ(\bbs)$ is as defined in~\eqref{eq:Jacobian}, and $\bbI$ is the identity matrix.

Clearly, analyzing the determinant in the integrand is crucial in volume analysis; we call it the \emph{volume integrand} in this paper.
When the determinant is expanded using the Leibniz formula, it becomes a polynomial of $\ep$, in the form of $1 + C(\bbs) \cdot \ep^h + \calO(\ep^{h+1})$ for some integer $h\ge 1$.
Thus, when $\ep$ is sufficiently small, the sign of $C(\bbs)$ dictates on whether the volume expands or contracts.

In our case, $\bbs$ refers to a cumulative payoff vector $(\bbp,\bbq)$.
Cheung and Piliouras~\cite{CP2019} showed that for the MWU update rule~\eqref{eq:MWU} in the dual space,
the volume integrand can be written as $1 + C_{(\bbA,\bbB)}(\bbp,\bbq) \cdot \ep^2 + \calO(\ep^4)$, where
\begin{equation}\label{eq:Cxy}
C_{(\bbA,\bbB)}(\bbp,\bbq) ~=~ -\sum_{j\in S_1}~\sum_{k\in S_2}~x_j(\bbp) \cdot y_k(\bbq) \cdot 
(A_{jk} - [\bbA\cdot \bby(\bbq)]_j) \cdot (B_{jk} - [\bbB\trans \cdot \bbx(\bbp)]_k).
\end{equation}
Note that $C_{(\bbA,\bbB)}(\bbp,\bbq)$ depends on the primal variables $\bbx(\bbp),\bby(\bbq)$ but not explicitly on $\bbp,\bbq$.
Thus, it is legitimate to refer to this value using the primal variables as input parameters to $C_{(\bbA,\bbB)}$, i.e., we can refer to its value by $C_{(\bbA,\bbB)}(\bbx,\bby)$ too.
Cheung and Piliouras~\cite{CP2019} showed the following lemma.% about $C(\bbx,\bby)$ in two-person zero-sum games.

\begin{lemma}\cite[Lemma 3, Section 4.1 and Appendix B]{CP2019} The following hold:\label{lem:Cxy}
\begin{enumerate}
\item When $\ep \le 1/4$, the update rule~\eqref{eq:MWU} in the dual space is injective.
\item In any two-person zero-sum game $(\bbA,-\bbA)$, at any point $(\bbx,\bby)\in \Delta$, $C_{(\bbA,-\bbA)}(\bbx,\bby) \ge 0$.
Indeed, $C_{(\bbA,-\bbA)}(\bbx,\bby)$ equals to the variance of the random variable $X$ such that
$X = (A_{jk} - [\bbA\bby]_j - [\bbA\trans\bbx]_k)$ with probability $x_j y_k$, for all $(j,k)\in S_1\times S_2$.
\item When $\ep < \min \left\{ 1/(32n^2m^2) , C(\bbp,\bbq) \right\}$, the volume integrand at point $(\bbp,\bbq)$ is lower bounded by $1+(C_{(\bbA,\bbB)}(\bbp,\bbq)-\ep)\ep^2$.
Thus, in~\eqref{eq:Liouville-discrete}, if $\overline{C} := \min_{(\bbp,\bbq)\in S} C_{\bbA,\bbB}(\bbp,\bbq) > 0$, then for all $0 < \ep \le \overline{C}$,
\[
\vol(S') ~\ge~ \left[ 1 + \left( \overline{C} - \ep \right) \ep^2 \right] \cdot \vol(S).
\]
\end{enumerate}
\end{lemma}

By the definition of $C_{(\bbA,\bbB)}$, it is straight-forward to see that
\begin{equation}\label{eq:coordination-negative-ZS}
C_{(\bbA,\bbA)}(\bbp,\bbq) ~=~ -C_{(\bbA,-\bbA)}(\bbp,\bbq).
\end{equation}
Thus, for any coordination game $(\bbA,\bbA)$, and for any $(\bbp,\bbq)\in \calD$, $C_{(\bbA,\bbA)}(\bbp,\bbq)\le 0$ due to Lemma~\ref{lem:Cxy} Part (2).

\subsection{Lyapunov Chaos}

In the study of dynamical systems, \textit{Lyapunov chaos} generally refers to following phenomenon in some systems:
a tiny difference in the starting points can yield widely diverging outcomes \emph{quickly}.
A classical measure of chaos is \textit{Lyapunov time}, defined as:
when the starting point is perturbed by a distance of tiny $\gamma$,
for how long will the trajectories of the two starting points remain within a distance of at most $2\gamma$.
Cheung and Piliouras~\cite{CP2019} showed that if the volume of a set increases at a rate of $\Omega((1+\beta)^t)$, %as it is evolved by a system,
its diameter increases at a rate of at least $\Omega((1+\beta/d)^t)$, where $d$ is the dimension of the system,
thus indicating that the Lyapunov time is at most $\calO(d/\beta)$.
\section{Unavoidability of MWU in Games}\label{sect:unavoidability}

The result in this section holds for MWU in general-sum games.
Recall the definition of $C_{(\bbA,\bbB)}(\bbp,\bbq)$ in Equation~\eqref{eq:Cxy},
and the discussion on extending the definition of the function $C$ to the primal space (i.e.~$C_{(\bbA,\bbB)}(\bbx,\bby)$) below Equation~\eqref{eq:Cxy}.
To avoid cluster, when the underlying game $(\bbA,\bbB)$ is clear from context, we write $C(\cdot)$ for $C_{(\bbA,\bbB)}(\cdot)$.

\begin{definition}\label{def:uncontrollable}
A set $U\subset \interior(\Delta)$ is called a \emph{primal open set} if there is an open set $U'$ in $\rr^{n+m}$, such that
$U = U'\cap \Delta$. A primal open set $U$ is \emph{uncontrollable} if $\inf_{(\bbx,\bby)\in U} C(\bbx,\bby) > 0$.
\end{definition}

The following is the \emph{unavoidability theorem}, the main theorem in this section.

\begin{theorem}\label{thm:unavoidable}
Let $U$ be an uncontrollable primal open set with $\inf_{(\bbx,\bby)\in U} C(\bbx,\bby) \ge \hC > 0$.
If the step-size $\ep$ in the update rule~\eqref{eq:MWU-primal} satisfies $\ep < \min \left\{ \frac{1}{32n^2m^2},\hC \right\}$,
then there exists a dense subset of points in $U$ such that the flow of each such point must eventually reach a point outside $U$.
\end{theorem}

Recall that one perspective to think about the unavoidability theorem is to consider $U$ as a collection of good points, while $\Delta \setminus U$ is the set of bad points that we want to \emph{avoid}.
We desire the game dynamic to stay within $U$ forever, so long as the starting point is in $U$.
The theorem then presents a negative property, which states that if $U$ is uncontrollable,
then there is a dense set of points in $U$ such that the game dynamic must eventually reach a point that we want to avoid.

In particular, when the underlying game is a zero-sum game, due to Lemma~\ref{lem:Cxy}, $\inf_{(\bbx,\bby)\in U} C(\bbx,\bby) \ge 0$ for \emph{any} $U$.
With some mild assumptions on $U$ and the underlying game,
it is foreseeable that the infimum will become strictly positive, for which Theorem~\ref{thm:unavoidable} is applicable.
For instance, if the zero-sum game is not trivial (see Definition~\ref{def:trivial}),
and $U$ collects all primal points $(\bbx,\bby)$ such that all $x_j,y_k\ge \delta$ for some fixed $\delta > 0$,
then the infimum is strictly positive due to Lemma~\ref{lem:Cxy} Part (2); see~\cite{CP2019} for a detailed explanation.
Informally speaking, for quite general scenarios, MWU in zero-sum game \emph{cannot} avoid bad states, regardless of what ``good'' or ``bad'' really mean.

\subsection{Proof of Theorem~\ref{thm:unavoidable}}

In Definition~\ref{def:uncontrollable}, we have defined uncontrollability of a set in the primal space.
In the dual space, the definition of uncontrollability is similar: an open set $V$ in the dual space is uncontrollable if $\inf_{(\bbp,\bbq)\in V} C(\bbp,\bbq) > 0$.

The following is the key lemma to proving Theorem~\ref{thm:unavoidable}.

\begin{lemma}\label{lem:key-in-dual}
Let $V$ be an uncontrollable open set in the dual space, with $\inf_{(\bbp,\bbq)\in V} C(\bbp,\bbq) \ge \hC > 0$.
Assume that the step-size $\ep$ in the update rule~\eqref{eq:MWU} satisfies $0< \ep < \min \left\{ \frac{1}{32n^2m^2},\hC \right\}$.
Let $S\subset V$ be a measurable set with positive volume, and let $S(t)$ be the flow of $S$ at time $t$.
Also, let
\[
d(S) ~:=~ \max \left\{ \max_{j\in S_1} \left\{ \max_{(\bbp,\bbq)\in S} p_j ~-~ \min_{(\bbp,\bbq)\in S} p_j \right\} ~,~ 
\max_{k\in S_2} \left\{ \max_{(\bbp,\bbq)\in S} q_k ~-~ \min_{(\bbp,\bbq)\in S} q_k \right\} \right\}.
\]
Then there exists a time $\tau$ with
\begin{equation}\label{eq:time-bound-unavoid}
\tau ~\le~ \max \left\{ \frac{d(S)}{2\ep} ~,~ \frac{8(n+m)}{(\hC-\ep)\ep^2} \ln \frac{4(n+m)}{(\hC-\ep)\ep^2} ~,~ \frac{4}{(\hC-\ep)\ep^2} \ln \frac{1}{\vol(S)} \right\},
\end{equation}
such that $S(\tau)$ contains a point which is \emph{not} in $V$.
\end{lemma}

\begin{pf}
We suppose the contrary, i.e., for all $\tau \le T$, $S(\tau) \subset V$, where $T$ will be specified later.
We analyze how the volume of $S(t)$ changes with $t$ using formula~\eqref{eq:Liouville-discrete}.
We rewrite it here:
\[
\vol(S(t+1)) ~=~ \int_{(\bbp,\bbq)\in S(t)} \det \left( \bbI + \ep \cdot \bbJ(\bbp,\bbq) \right) \,\mathsf{d}V.
\]
By Lemma~\ref{lem:Cxy}, if $S(t)\subset V$, then the above inequality yields 
$\vol(S(t+1)) ~\ge~ \vol(S(t)) \cdot \left( 1+(\hC-\ep)\ep^2 \right)$, and hence
\begin{equation}\label{eq:exp-LB}
\forall t\le T+1,~~~~\vol(S(t)) ~\ge~ \vol(S) \cdot \left( 1+(\hC-\ep)\ep^2 \right)^t.
\end{equation}

On the other hand, observe that in the update rule~\eqref{eq:MWU}, each variable is changed by a value in the interval $[-\ep,\ep]$ per time step,
since every entry in $\bbA,\bbB$ is in the interval $[-1,1]$.
Consequently, the range of possible values for each variable in $S(t)$ lies within an interval of length at most $d(S) + 2\ep t$,
and hence $S(t)$ is a subset of a hypercube with side length $d(S) + 2\ep t$. Therefore,
\begin{equation}\label{eq:poly-UB}
\forall t\le T+1,~~~~\vol(S(t)) ~\le~ \left( d(S) + 2\ep t \right)^{n+m}.
\end{equation}

Note that the lower bound in~\eqref{eq:exp-LB} is exponential in $t$, while the upper bound in~\eqref{eq:poly-UB} is polynomial in $t$.
Intuitively, it is clear that the two bounds cannot be compatible for some large enough $T$, and hence a contradiction.
The rest of this proof is to derive how large $T$ should be.
Precisely, we seek $T$ such that
\[
\left( d(S) + 2\ep T \right)^{n+m} ~<~ \vol(S) \cdot \left( 1+(\hC-\ep)\ep^2 \right)^T.
\]
First, we impose that $T\ge d(S)/(2\ep)=:T_1$. Taking logarithm on both sides, to satisfy the above inequality, it suffices that
\[
(n+m) \ln (4\ep T) ~<~ \frac{T \cdot (\hC-\ep)\ep^2}{2} + \ln (\vol(S)).
\]
Since $4\ep \le 1$, it suffices that
\[
(\hC-\ep)\ep^2T - 2(n+m) \ln T ~>~ 2\cdot \ln \frac{1}{\vol(S)}.
\]

Next, observe that when $T\ge \frac{8(n+m)}{(\hC-\ep)\ep^2} \ln \frac{4(n+m)}{(\hC-\ep)\ep^2} =: T_2$,
we have $(\hC-\ep)\ep^2T - 2(n+m) \ln T \ge (\hC-\ep)\ep^2T/2$. (We will explain why in the next paragraph.)
Then it is easy to see that $T\ge \frac{4}{(\hC-\ep)\ep^2} \ln \frac{1}{\vol(S)} =: T_3$ suffices. Overall, we need $T = \max \{T_1,T_2,T_3\}$.

Lastly, we explain why the inequality in the last paragraph holds. Observe that it is equivalent to $\frac{T}{\ln T} \ge \frac{4(n+m)}{(\hC-\ep)\ep^2} =: \gamma$.
Then it suffices to know that $\frac{T}{\ln T}$ is an increasing function of $T$ when $T\ge 3$, and
\[
\frac{T_2}{\ln T_2} ~=~ \frac{2\gamma \ln \gamma}{\ln 2 + \ln \gamma + \ln\ln \gamma} ~\ge~ \frac{2\gamma \ln \gamma}{2\ln \gamma} ~=~ \gamma,
\]
where the only inequality sign in the above expression holds because $\ln \gamma \ge \ln \ln \gamma + \ln 2 > 0$ when $\gamma\ge 3$.
\end{pf}

The following proposition is straight-forward.

\begin{proposition}\label{pr:tD}
If $U$ is a primal open set, then $\tGinv(U)$ is an open and unbounded subset in $\calD$.
\end{proposition}

\begin{pfof}{Theorem~\ref{thm:unavoidable}}
Let $U'$ denote the set of points in $U$ which, when taken as a starting point, will eventually reach a point outside $U$.
Suppose the theorem does not hold, i.e., $U'$ is not dense. %the set of points in the primal set $U$ which eventually reach a point outside $U$ is not dense.
Then we can find a primal open set $B\subset U$ such that its flow must stay in $U$ forever. %In particular, $\tP^{-1}(B)\subset U$.

Let $V := \tGinv(U)$ and $S' := \tGinv(B)$. 
Due to the discussion immediately after Equation~\eqref{eq:Cxy} and the assumption that $U$ is uncontrollable in the primal space, $V$ is uncontrollable in the dual space.
On the other hand, $S'$ is open and unbounded due to Proposition~\ref{pr:tD}.
But it is easy to find a subset $S\subset S'$ which is open and bounded. Thus, $S$ has positive and finite volume.
We apply Lemma~\ref{lem:key-in-dual} with the sets $V,S$ given above,
to show that using update rule~\eqref{eq:MWU}, the flow of $S$ at some time $\tau$ contains a point $(\bbp^\tau,\bbq^\tau)\notin V$.
By definition of $V$, $\tG(\bbp^\tau,\bbq^\tau)\notin U$.

Let $(\bbp^0,\bbq^0)$ denote a point in $S$ such that its flow at time $\tau$ is $(\bbp^\tau,\bbq^\tau)$.
Since $S$ is a subset of $S'$, $\tG(\bbp^0,\bbq^0)\in B$.
Due to the equivalence between the primal update rule~\eqref{eq:MWU-primal} and the dual update rule~\eqref{eq:MWU},
we can conclude that when $\tG(\bbp^0,\bbq^0) \in B$ is used as the starting point of the primal update rule~\eqref{eq:MWU-primal},
at time $\tau$ its flow is $\tG(\bbp^\tau,\bbq^\tau)$ which is not in $U$, a contradiction.
\end{pfof}
\section{Extremism of MWU in Zero-Sum Games}\label{sect:extremism}

Here, we focus on MWU in zero-sum game.
\cite{BP2018} and \cite{Cheung2018} showed that the dynamic converges to the boundary of $\Delta$
and fluctuates bizarrely near the boundary by using a potential function argument.
However, the potential function has value $+\infty$ at every point on the boundary so it cannot be distinctive there,
and hence it cannot provide any useful insight on how the dynamic behaves near the boundary.
In general, the behaviors near boundary can be highly unpredictable, as suggested by the ``chaotic switching'' phenomenon found by~\cite{AC2010},
although more regular (yet still surprising) patterns were found in lower-dimensional systems~\cite{Gau1992}.

In~\cite{BP2018,Cheung2018}, a central \emph{discouraging} message is convergence towards boundary of $\Delta$ is inevitable
even when the underlying zero-sum game has a fully-mixed Nash equilibrium.
What can we still hope for after this? Will $(\bbx^t,\bby^t)$ remain within a somewhat \emph{reasonable} range around the Nash equilibrium forever?
We answer the latter question with a strikingly general negative answer for almost all zero-sum games, with the two theorems below.

\begin{definition}
The \emph{extremal domain with threshold $\delta$} consists of all points $(\bbx,\bby)$ such that
each of $\bbx,\bby$ has exactly one entry of value at least $1-\delta$.
\end{definition}

\begin{theorem}\label{thm:extremism}
Let $(\bbA,-\bbA)$ be a two-person zero-sum game. Suppose the following:
\begin{enumerate}
\item[(A)] Every $2\times 2$ sub-matrix of $\bbA$ is non-trivial. Let $\alpha_1 > 0$ denote the minimum distance from triviality of all
$2\times 2$ sub-matrices of $\bbA$. (Recall the distance measure~\eqref{eq:distance-from-trivial}.)
\item[(B)] No two entries in the same row or the same column have exactly the same value.
Let $\alpha_2 > 0$ be the minimum difference between any two entries of $\bbA$ in the same row or the same column.
\end{enumerate}
Let $N := \max\{n,m\}$. For any $0 < \delta < \alpha_2/4$, if both players use MWU with step-size $\ep$ satisfying
$
0< \ep <  \min \left\{ \frac{1}{32n^2 m^2} ~,~ \frac {(\alpha_1)^2}{18} \cdot \left(\frac{\delta}{N-1}\right)^{8(N-1)/(\alpha_2-4\delta)+2} \right\},
$
then there exists a dense subset of points in $\interior(\Delta)$,
such that the flow of each such point must eventually reach the extremal domain with threshold $\delta$.
\end{theorem}

\begin{theorem}\label{thm:extremal-io}
Let $v$ denote the game value of the zero-sum game $(\bbA,-\bbA)$.
In addition to the conditions required in Theorem~\ref{thm:extremism},
if (i) $\min_{j\in S_1,k\in S_2} |A_{jk} - v| \ge r > 0$, and (ii) $6\ep + 4\delta \le r$,
then there exists a dense subset of points in $\interior(\Delta)$,
such that the flow of each such point visits and leaves extremal domain with threshold $\delta$ infinitely often. 
\end{theorem}

To see the power of the Theorem~\ref{thm:extremism}, consider a zero-sum game with a fully-mixed Nash Equilibrium.
The theorem implies that in \emph{any arbitrarily small open neighbourhood} of the Nash equilibrium,
there exists a starting point such that its flow will eventually reach a point where each player concentrates her game-play on only one strategy.
We call this \emph{extremism of game-play}, since both players are single-minded at this point:
they are concentrating on one strategy and essentially ignoring all the other available options.

There are two assumptions on the matrix $\bbA$. If the matrix is to be drawn uniformly randomly from the space $[-1,+1]^{n\times m}$,
the random matrix satisfies assumptions (A) and (B) almost surely.
Unfortunately, the classical Rock-Paper-Scissors game is a zero-sum game which does not satisfy assumption (A) in Theorem~\ref{thm:extremism}, and thus the theorem is not applicable.
In Appendix~\ref{app:RPS}, we provide a separate proof which shows similar result to Theorem~\ref{thm:extremism} for this specific game.

\subsection{Proof Sketch of Theorem~\ref{thm:extremism}}

The full proofs of the two theorems are deferred to Appendix~\ref{app:extremism}. Here, we give high-level description of the proof of Theorem~\ref{thm:extremism}.

We first define a family of primal open sets in $\interior(\Delta)$.
Let $\calEd_{a,b}$ be the collection of all points $(\bbx,\bby)$, such that 
at least $a$ entries in $\bbx$ are larger than $\delta$, and at least $b$ entries in $\bby$ are larger than $\delta$.
The first step is to use condition (A) to show that for any $1/3 > \kappa > 0$,
\begin{equation}\label{eq:calE-uncontrol}
\calE^\kappa_{2,2}~\text{is an uncontrollable primal set with}~~~\inf_{(\bbx,\bby)\in \calE^\kappa_{2,2}} C(\bbx,\bby) \ge \kappa^2 (\alpha_1)^2 / 2.
\end{equation}

Then we can apply Theorem~\ref{thm:unavoidable} to show that for any sufficiently small step-size $\ep$, there exists
a dense subset of points in $\calE^\kappa_{2,2}$ such that the flow of each such point must eventually reach a point outside $\calE^\kappa_{2,2}$.
Let $(\hbbx,\hbby)$ denote the reached point outside $\calE^\kappa_{2,2}$.
At $(\hbbx,\hbby)$, one of the two players, which we assume to be Player 1 without loss of generality,
concentrates her game-play on only one strategy, which we denote by strategy $\hj$.
We have: for any $j\neq \hj$, $\hat{x}_j \le \kappa$, and hence $\sum_{j\in S_1\setminus\{\hj\}} \hat{x}_j \le (N-1)\kappa$.

When we pick $(N-1)\kappa \ll \delta$, where $\delta$ is the quantity specified in Theorem~\ref{thm:extremism},
after the dynamic reaches $(\hbbx,\hbby)$, the total probability of choosing any strategy other than $\hj$ by Player 1 will be at most $\delta$
for a long period of time --- this is true because that total probability can increase by a factor of at most $\exp(2\ep)$ per time step.
In other words, we may think that the game essentially becomes an $1\times m$ sub-game of $(\bbA,-\bbA)$ during this long period of time.

We then show that during the long period of time, \emph{no matter what $\hbby$ is}, the game-play of Player 2 must become concentrating on one strategy too.
Naively, one might think that this strategy ought to be $k$ which maximizes $-A_{\hj k}$, which is the dominant strategy of Player 2 in the $1\times m$ sub-game.
However, if $\hat{y}_k$ is tiny, this might not be true. To reach the conclusion, we have to use a technical lemma, Lemma~\ref{lem:single-minded} in Appendix~\ref{app:extremism}.
\section{Continuous Analogue of OMWU}\label{sect:cont-analog}

As the update rule~\eqref{eq:OptMD} at time $t+1$ depends on the past updates at times $t$ and $t-1$,
at first sight it might seem necessary to perform volume analysis in the product space $\Delta \times \Delta$ that contains $((\bbp_t,\bbq_t),(\bbp_{t-1},\bbq_{t-1}))$.
However, this raises a number of technical difficulties.
First, since the initialization sets $(\bbp_1,\bbq_1)$ as a function of $(\bbp_0,\bbq_0)$,
the initial set has to lie in a proper manifold in $\Delta \times \Delta$,
thus it has zero Lebesgue measure w.r.t.~$\Delta \times \Delta$, making volume analysis useless, as the volume must remain zero when the initial set is of measure zero.
Second, even if we permit $\bbp_1,\bbq_1$ to be unrelated to $\bbp_0,\bbq_0$ so that we can permit an initial set with positive measure,
the OMWU update rule is not of the same type that is presumed by the formula~\eqref{eq:Liouville-discrete}.
We will need to use the more general form of integration by substitution, and the volume integrand there
will \emph{not} be of the form $\bbI + \ep \cdot \bbJ$, hence the determinant is not a polynomial of $\ep$ with constant term $1$.
This imposes huge difficulty in analysis, forbidding us to present a clean volume analysis as was done in~\cite{CP2019}.
To bypass the issues, we first derive a continuous analogue of OMWU in games as an ODE system, which will permit us to have a clean volume analysis.

\subsection{Continuous Analogue of OMWU in General Contexts}\label{sect:cont-analog-general}

%We first give some high level discussion about the OMWU update rule~\eqref{eq:OptMD} in the dual space.
We focus on Player 1 who uses the OMWU update rule~\eqref{eq:OptMD} in the dual space.
To set up for the most general context, we replace $\bbA \cdot \bby(\bbq^t)$ by $\bbu(t)$,
which represents the utility (or payoff) vector at time $t$. We assume $\bbu(t)$ is $C^2$-differentiable.
%Note that this replacement of notation implicitly assumes that $\bbu$ is not depending on $\bbp,\bbq$ anymore.
We rewrite the rule as below:
\begin{equation}\label{eq:OptMD-rewrite}
\frac{\bbp^{t+1} ~-~ \bbp^t}{\ep} ~=~  \bbu(t) ~+~ \ep \cdot\frac{\bbu(t) - \bbu(t-1)}{\ep}.
\end{equation}
Recall that for any smooth function $f:\rr\ra\rr$,
its first derivative is $\lim_{\ep\ra 0} (f(x+\ep)-f(x))/\ep$.
For readers familiar with Euler discretization and finite-difference methods,
the above discrete-time rule naturally motivates the following differential equation, where for any variable $v$, $\dot v \equiv \difft{v}$:
\begin{equation}\label{eq:DE-OptMD-general}
\dot \bbp ~=~ \bbu ~+~ \ep \cdot \dot{\bbu}.
\end{equation}

To numerically simulate~\eqref{eq:DE-OptMD-general}, we should take into account various informational constraints:
\begin{itemize}[leftmargin=0.2in]
\item If the function $\bbu$ is explicitly given and it is a simple function of time (e.g.~a polynomial),
the function $\dot{\bbu}$ can be explicitly computed.
Euler method with step-size $\Delta t = \ep$ is the update rule
$
\bbp(t+\ep) = \bbp(t) + \ep \cdot \bbu(t) + \ep^2 \cdot \dot{\bbu}(t).
$
\item However, in some scenarios, $\bbu$ is a rather complicated function of $t$,
so computing explicit formula for $\dot{\bbu}$ might not be easy.
Yet, we have full knowledge of values of $\bbu(0),\bbu(\Delta t),\bbu(2\cdot \Delta t),\bbu(3\cdot \Delta t),\cdots$.
Then a common approach to approximately compute $\dot{\bbu}(N\cdot \Delta t)$ is to use the central finite-difference method: 
$
\dot{\bbu}(N\cdot \Delta t) ~=~ \frac{\bbu((N+1)\cdot \Delta t) - \bbu((N-1)\cdot \Delta t)}{2\cdot \Delta t} ~+~ \calO((\Delta t)^2).
$
Euler method with step-size $\Delta t = \ep$ which makes use of the above approximation gives the update rule
$
\bbp(t+\ep) = \bbp(t) + \ep \cdot \bbu(t) + \ep \cdot \frac{\bbu(t+\ep) - \bbu(t-\ep)}{2}.
$
%with local error $\calO(\ep^2) + \calO(\ep^3) = \calO(\ep)$.
\item Even worse, in the context of online learning or game dynamics,
at time $N\cdot \Delta t$, the players have only observed
$\bbu(0),\bbu(\Delta t),\bbu(2\cdot \Delta t),\cdots,\bbu(N\cdot \Delta t)$,
but they do not have any knowledge on the \emph{future} values of $\bbu$.
Due to the more severe constraint on information,
we have to settle with the backward finite-difference method to approximately compute $\dot{\bbu}(N\cdot \Delta t)$:
$
\dot{\bbu}(N\cdot \Delta t) ~=~ \frac{\bbu(N\cdot \Delta t) - \bbu((N-1)\cdot \Delta t)}{\Delta t} ~+~ \calO(\Delta t),
$
which has a higher-order error when compared with the central finite-difference method.
Euler method with step-size $\Delta t = \ep$ which makes use of the above approximation gives the rule~\eqref{eq:OptMD-rewrite}, by identifying $\bbp(t+\ep)$ as $\bbp^{t+1}$.
Due to an error that occurs when we approximate $\dot{\bbu}$ as above,
\begin{equation}\label{eq:online-Euler-error}
\ep \cdot \bbu(t) + \ep\cdot (\bbu(t)-\bbu(t-1)) ~=~ \ep \left[\bbu(t) + \ep \cdot \dot{\bbu}(t)\right] + \calO(\ep^3),
\end{equation}
where the LHS is the quantity $\bbp^{t+1} - \bbp^t$ in the OMWU update rule~\eqref{eq:OptMD-rewrite},
and the first term in the RHS is the standard Euler discretization of~\eqref{eq:DE-OptMD-general}.
\end{itemize}
%The local errors for the three cases can all be eventually expressed as $\calO(\ep^2)$.
%One common source of local error term occurs when we approximate any $\bbu(t+\Delta t)$ for $\Delta t\in [0,\ep]$ by $\bbu(t)$.
%Modulo this main error term, the first and second cases both have remaining errors of $\calO(\ep^3)$,
%while the third case has worse remaining error of $\calO(\ep^2)$.

%The local truncation errors for the three cases are all $\calO(\ep^2)$, which is standard for Euler method.
%However, the hidden constants in the big-Oh notation vary.
%We conclude the above discussion with the following proposition.

\begin{proposition}\label{prop:OptMD-discretization}
From differential equation~\eqref{eq:DE-OptMD-general}, when only \emph{online value oracle} for a $C^2$-differentiable function $\bbu$ is given,
the discrete-time update rule~\eqref{eq:OptMD} is obtained by first using backward finite-difference method with step-size $\ep$ to approximate $\dot{\bbu}$,
and then applying the Euler discretization method with step-size $\ep$. Also, Equation~\eqref{eq:online-Euler-error} holds.
\end{proposition}

\vspace*{-0.1in}

While we \emph{heuristically derived} the ODE system~\eqref{eq:DE-OptMD-general} from OMWU with $\ep$ being the step-size,
but after it is derived, $\ep$ becomes a parameter of a parametrized family of learning dynamics.
%In the ODE system~\eqref{eq:DE-OptMD-general}, 
When this parameter $\ep$ is zero, system~\eqref{eq:DE-OptMD-general} recovers the Replicator Dynamics. %, which does not depend on $\dot{\bbu}$.
When $\ep > 0$, it reduces the increment of $p_j$ if $\dot u_j$ is negative.
It can be interpreted as a common learning behavior (e.g.~in financial markets), which is mainly depending on the payoffs, but also having a tuning which depends on the \emph{trend} of payoffs.
There is nothing to stop us from having a negative $\ep$, although it is not clear in what contexts such learning dynamics are motivated.

%While we \emph{heuristically derived} the ODE system from OptMD by assuming $\ep$ is a small quantity, but once after it is derived,
%we can view the system as a family of parametrized learning dynamics with an \emph{unrestricted} parameter $\ep$: $\ep$ can be positive or negative, and its magnitude can be large (say $1/2$).
%It is legitimate to consider online Euler discretization with a positive step-size $\alpha$ which can be different from $\ep$, which yields the update rule
%$\bbp^{t+1} =  \bbp^t + (\alpha+\ep) \cdot \bbu(t) - \ep \cdot \bbu(t-1)$.

\subsection{Continuous Analogue of OMWU in General-Sum Games}\label{sect:cont-analog-game}
Next, we use~\eqref{eq:DE-OptMD-general} to derive a system of differential equations for OMWU in general-sum games.
In these and also many other learning contexts, $\bbu,\dot{\bbu}$ depend on the driving variables $\bbp,\bbq$.
In~\eqref{eq:DE-OptMD-general}, for Player 1, we replace $\bbu(t)$ by $\bbA \cdot \bby(\bbq^t)$. By the chain rule,
\[
\difft{p_j} ~=~ [\bbA \cdot \bby(\bbq)]_j + \ep \cdot \difft{[\bbA \cdot \bby(\bbq)]_j} ~=~ [\bbA \cdot \bby(\bbq)]_j + \ep \cdot \sum_{k\in S_2} \frac{\partial [\bbA\cdot \bby(\bbq)]_j}{\partial q_k} \cdot \difft{q_k}.
\]
%and hence $\dot{\bbu}(t) = \difft{\bbA \cdot \bby(\bbq^t)}$.
%By doing so, we express $\dot{\bbp}$ in terms of $\dot{\bbq}$ on Player 1's perspective;
%analogously, we express $\dot{\bbq}$ in terms of $\dot{\bbp}$ on Player 2's perspective.
%These constitute to a set of recursive formulae.
%Precisely, 
Recall from~\cite[Equation (7)]{CP2019} that $\frac{\partial [\bbA\cdot \bby(\bbq)]_j}{\partial q_k} = y_k(\bbq) \cdot (A_{jk} - [\bbA \cdot \bby(\bbq)]_j)$. Thus,
\begin{align}
\difft{p_j} &= [\bbA \cdot \bby(\bbq)]_j + \ep \sum_{k\in S_2} y_k(\bbq) \cdot (A_{jk} - [\bbA \cdot \bby(\bbq)]_j) \cdot \difft{q_k}.\label{eq:recur-p}\\
\text{Analogously,}~~~~\difft{q_k}
&= [\bbB\trans \cdot \bbx(\bbp)]_k + \ep \sum_{j\in S_1} x_j(\bbp) \cdot (B_{jk} - [\bbB\trans \cdot \bbx(\bbp)]_k) \cdot \difft{p_j}.\label{eq:recur-q}
~~~~~~~~~~~~~~~~~
\end{align}

Formally, the above two formulae, which are in a recurrence format, have \emph{not} yet formed an ODE system. To settle this issue,
in Appendix~\ref{app:cont-analogue-GS-games}, we show that when $\ep$ is small enough, they can be reduced to a standard ODE system of the form
\[
\left( \difft{\bbp} , \difft{\bbq} \right)\trans = %\left( \bbI + \sum_{\ell=1}^\infty \ep^\ell \cdot \bbM(\bbp,\bbq)^\ell \right) 
\left[ \bbI - \ep \bbM(\bbp,\bbq) \right]^{-1}\cdot \bbv(\bbp,\bbq)
\]
for some matrix $\bbM(\bbp,\bbq)$ and vector $\bbv(\bbp,\bbq)$.
This then formally permits us to use~\eqref{eq:recur-p} and~\eqref{eq:recur-q} in the analysis below, as is standard in formal power series when dealing with generating functions.
\section{Volume Analysis of OMWU in Games}\label{sect:volume-contract}

Iterating the recurrence~\eqref{eq:recur-p} and~\eqref{eq:recur-q} yields the following system.
\begin{align}
\difft{p_j} &~=~ [\bbA \cdot \bby(\bbq)]_j ~+~ \ep \sum_{k\in S_2} y_k(\bbq)\cdot (A_{jk} - [\bbA \cdot \bby(\bbq)]_j) \cdot [\bbB\trans \cdot \bbx(\bbp)]_k ~+~ \calO(\ep^2); \nonumber\\
\difft{q_k} &~=~ [\bbB\trans \cdot \bbx(\bbp)]_k ~+~ \ep \sum_{j\in S_1} x_j(\bbp) \cdot (B_{jk} - [\bbB\trans \cdot \bbx(\bbp)]_k) \cdot [\bbA \cdot \bby(\bbq)]_j ~+~ \calO(\ep^2).\label{eq:explicit-dp-dq-dt}
\end{align}

Proposition~\ref{prop:OptMD-discretization} establishes that in general contexts,~\eqref{eq:OptMD} is the online Euler discretization of the differential equation~\eqref{eq:DE-OptMD-general}.
As a special case in games,~\eqref{eq:OptMD} is the online Euler discretization of the recurrence system~\eqref{eq:recur-p} and~\eqref{eq:recur-q}.
Via Equations~\eqref{eq:explicit-dp-dq-dt} and~\eqref{eq:online-Euler-error}, we rewrite~\eqref{eq:OptMD} as
\begin{align*}
p_j^{t+1} &=~ p_j^t + \ep [\bbA \cdot \bby(\bbq^t)]_j ~+~ \ep^2 \sum_{k\in S_2} y_k(\bbq^t)\cdot (A_{jk} - [\bbA \cdot \bby(\bbq^t)]_j) \cdot [\bbB\trans \cdot \bbx(\bbp^t)]_k + \calO(\ep^3); \nonumber\\
q_k^{t+1} &=~ q_k^t + \ep [\bbB\trans \cdot \bbx(\bbp^t)]_k ~+~ \ep^2 \sum_{j\in S_1} x_j(\bbp^t) \cdot (B_{jk} - [\bbB\trans \cdot \bbx(\bbp^t)]_k) \cdot [\bbA \cdot \bby(\bbq^t)]_j + \calO(\ep^3).
\end{align*}

Update rule~\eqref{eq:OptMD} can be implemented easily by the players in distributed manner, but it is hard to be used for volume analysis.
The above update rule %~\eqref{eq:OptMD-convenient-for-analysis} 
\emph{cannot} be implemented by the players in distributed manner, since
Player 1 does not know the values of $y_k$ and $[\bbB\trans \cdot \bbx(\bbp)]_k$. %, and in many common models she also has no knowledge of individual $A_{jk}$ values too.
However, it will permit us to perform a clean volume analysis, since its RHS involves only $\bbp^t,\bbq^t$ but not $\bbp^{t-1},\bbq^{t-1}$.
We will show that when we ignore the $\calO(\ep^3)$ terms %in~\eqref{eq:OptMD-convenient-for-analysis} 
and perform volume analysis as described in Section~\ref{sect:volume-analysis-prelim},
the volume integrand %in~\eqref{eq:OptMD-convenient-for-analysis} 
is of the format $1+C'\ep^2+ \calO(\ep^3)$.
Thus, taking the ignored terms into account does not affect the crucial $C'\ep^2$ term which dictates volume change.

For the moment, we ignore the $\calO(\ep^3)$ terms. % in~\eqref{eq:OptMD-convenient-for-analysis}.
To use~\eqref{eq:Liouville-discrete} for computing volume change, we need to derive $\ep\cdot \bbJ(\bbp,\bbq)$ in the volume integrand:
\begin{align*}
\forall j_1,j_2\in S_1,~~~~~~& \ep J_{j_1 j_2} ~=~ \ep^2 \sum_{k\in S_2} y_k(\bbq) \cdot (A_{j_1k} - [\bbA \cdot \bby(\bbq)]_{j_1}) \cdot x_{j_2}(\bbp) \cdot (B_{j_2k} - [\bbB\trans \cdot \bbx(\bbp)]_k)~;\\
\forall k_1,k_2\in S_2,~~~~~~& \ep J_{k_1 k_2} ~=~ \ep^2 \sum_{j\in S_1} x_j(\bbp) \cdot (B_{jk_1} - [\bbB\trans \cdot \bbx(\bbp)]_{k_1}) \cdot y_{k_2}(\bbq) \cdot (A_{jk_2} - [\bbA \cdot \bby(\bbq)]_j)~;\\
\forall j\in S_1, k\in S_2,~~& ~\ep J_{jk}~~=~ \ep \cdot y_k(\bbq)  \cdot (A_{jk} - [\bbA \cdot \bby(\bbq)]_j) ~+~ \calO(\ep^2)~;\\
\forall k\in S_2, j\in S_1,~~& ~\ep J_{kj}~~=~ \ep \cdot x_j(\bbp) \cdot (B_{jk} - [\bbB\trans \cdot \bbx(\bbp)]_k) ~+~ \calO(\ep^2)~.
\end{align*}
With the above formulae, we expand $\det(\bbI + \ep \cdot \bbJ(\bbp,\bbq))$ via the Leibniz formula. The determinant
is of the form $1+C'(\bbp,\bbq)\cdot \ep^2 + \calO(\ep^3)$, where $C'(\bbp,\bbq)$ is the coefficient of $\ep^2$ in the expression
\[
\sum_{j\in S_1} \ep J_{jj} ~+~ \sum_{k\in S_2} \ep J_{kk} ~-~ \sum_{\substack{j\in S_1\\k\in S_2}}~(\ep J_{jk})(\ep J_{kj}).
\]
A straight-forward arithmetic shows the above expression equals to $-\ep^2 \cdot C_{(\bbA,\bbB)}(\bbp,\bbq) + \calO(\ep^3)$, and hence
\begin{equation}\label{eq:integrand-OptMD}
\det(\bbI + \ep \cdot \bbJ(\bbp,\bbq)) ~=~ 1 ~-~ C_{(\bbA,\bbB)}(\bbp,\bbq) \cdot \ep^2 ~+~ \calO(\ep^3).
\end{equation}
\subsection{OMWU in Coordination Games is Lyapunov Chaotic in the Dual Space}\label{sect:Lyapunov-coordination}

At this point, it is important to address the similarity of MWU in zero-sum games $(\bbA,-\bbA)$ and OMWU in coordination games $(\bbA,\bbA)$.
Recall from~\cite{CP2019} that the volume integrand for the former case is 
\[
1+C_{(\bbA,-\bbA)}(\bbp,\bbq)\cdot \ep^2 + \calO(\ep^4),
\]
while by~\eqref{eq:integrand-OptMD} and~\eqref{eq:coordination-negative-ZS},
the volume integrand for the latter case is 
\[
1-C_{(\bbA,\bbA)}(\bbp,\bbq)\cdot \ep^2 + \calO(\ep^3) = 1+C_{(\bbA,-\bbA)}(\bbp,\bbq)\cdot \ep^2 + \calO(\ep^3).
\]
When $\ep$ is the sufficiently small, their volume-change behaviors are almost identical.
Using~\eqref{eq:calE-uncontrol}, we can deduce all the Lyapunov chaos, unavoidability and extremism results in Section~\ref{sect:unavoidability} for OMWU in coordination games.
We also have volume contraction results for OMWU in zero-sum game and MWU in coordination game,
which are stated formally in Theorems~\ref{thm:OMD-volume-contract-zerosum-discrete} and~\ref{thm:MWU-volume-contract-coordination-discrete} in Appendix~\ref{app:vol-analysis-OptMD}.

%Next, we recall the fact $C(\bbp,\bbq) \le 0$ in coordination game. 
%In the same manner as proving Theorem~\ref{thm:OMD-volume-expand-coordination}, we have the following theorem.

\begin{theorem}\label{thm:OMD-volume-expand-coordination-discrete}
Suppose the underlying game is a non-trivial coordination game $(\bbA,\bbA)$ and the parameter $\alpha_1$ as defined in Theorem~\ref{thm:extremism} is strictly positive.
For any $1/2 > \delta > 0$, for any sufficiently small $0<\ep\le \bar{\ep}$ where the upper bound depends on $\delta$,
and for any set $S=S(0)\subset \tGinv(\calEd_{2,2})$ in the dual space, 
if $S$ is evolved by the OMWU update rule~\eqref{eq:OptMD} and if its flow remains a subset of $\tGinv(\calEd_{2,2})$ for all $t\le T-1$, then
$\vol(S(T)) ~\ge~ \left( 1 + \frac{\ep^2 \delta^2 (\alpha_1)^2}{4} \right)^T \cdot \vol(S).$
Consequently, the system is Lyapunov chaotic within $\tGinv(\calEd_{2,2})$ of the dual space, with Lyapunov time $\calO((n+m)/(\ep^2 \delta^2 (\alpha_1)^2))$.
\end{theorem}

\subsection{Negative Consequences of Volume Expansion of OMWU in Coordination Game}

In Sections~\ref{sect:unavoidability} and~\ref{sect:extremism}, the unavoidability and extremism theorems are proved largely due to volume expansion;
For the extremism theorems, it requires some additional arguments that seem specific to MWU,
which comprise of Lemma~\ref{lem:single-minded} and Step 3 in the proof of Theorem~\ref{thm:extremism}, both in Appendix~\ref{app:extremism}).
But a careful examination of the proof of Lemma~\ref{lem:single-minded} and the Step 3 finds these additional arguments work for OMWU too (with very minor modifications).
Thus, the unavoidability and extremism theorems hold for OMWU too, after suitably modifying the condition needed for volume expansion, and the upper bounds on the step-sizes.

Suppose a coordination game has a non-pure Nash equilibrium (i.e.~a Nash equilibrium $(\bbx^*,\bby^*)$ in which the supports of $\bbx^*,\bby^*$ are both of size at least $2$).
By Theorem~\ref{thm:extremism} (the OMWU analogue), for any tiny open ball $B$ around the equilibrium, there is a dense subset of points in $\interior(\Delta)\cap B$ such that
the flow of this point eventually reaches close to an extremal point.
In other words, there are points arbitrarily close to the equilibrium with their flows reaching extremal points, i.e.~the
flows not only move away from the equilibrium \emph{locally}, but they move away for a big distance.
This kind of \emph{global instability} result can be applied quite broadly, as many coordination games have non-pure Nash equilibrium.
In the standard coordination game $(\bbA,\bbA)$ where 
$\bbA = \left[\begin{smallmatrix}
1 & 0 \\
0 & 1
\end{smallmatrix}\right]$.,
the game has three Nash equilibria, namely $((1,0),(1,0))$, $((0,1),(0,1))$ and $((1/2,1/2),(1/2,1/2))$.
The latest one is a non-pure Nash equilibrium.

Another example is the following generalization.
Consider a two-player coordination game where each player has $n$ strategies.
Suppose that when both players choose strategy $i$, they both earn \$$A_i$, and otherwise they both lose \$$Z$, where $A_i>0$ and $Z\ge 0$.
Then the game has a non-pure Nash equilibrium $(\bbx^*,\bbx^*)$, where $x^*_i = \frac{1}{A_i+Z} \left/ \left(\sum_j \frac{1}{A_j+Z}\right) \right.$, which is strictly positive for all $i$.

\section*{Acknowledgments}

We thank several anonymous reviewers for their suggestions, which help to improve the readability of this paper from its earlier version.
Yun Kuen Cheung and Georgios Piliouras acknowledge 
AcRF Tier 2 grant 2016-T2-1-170, grant PIE-SGP-AI-2018-01, NRF2019-NRF-ANR095 ALIAS grant and NRF 2018 Fellowship NRF-NRFF2018-07.

\bibliographystyle{plain}
\bibliography{ref,refer,refer2}

\newpage

\begin{figure}[htp]
\begin{center}
\includegraphics[scale=0.385]{MWU-zerosum-500-28}~
\includegraphics[scale=0.385]{OMWU-zerosum-500-28}\\
\includegraphics[scale=0.385]{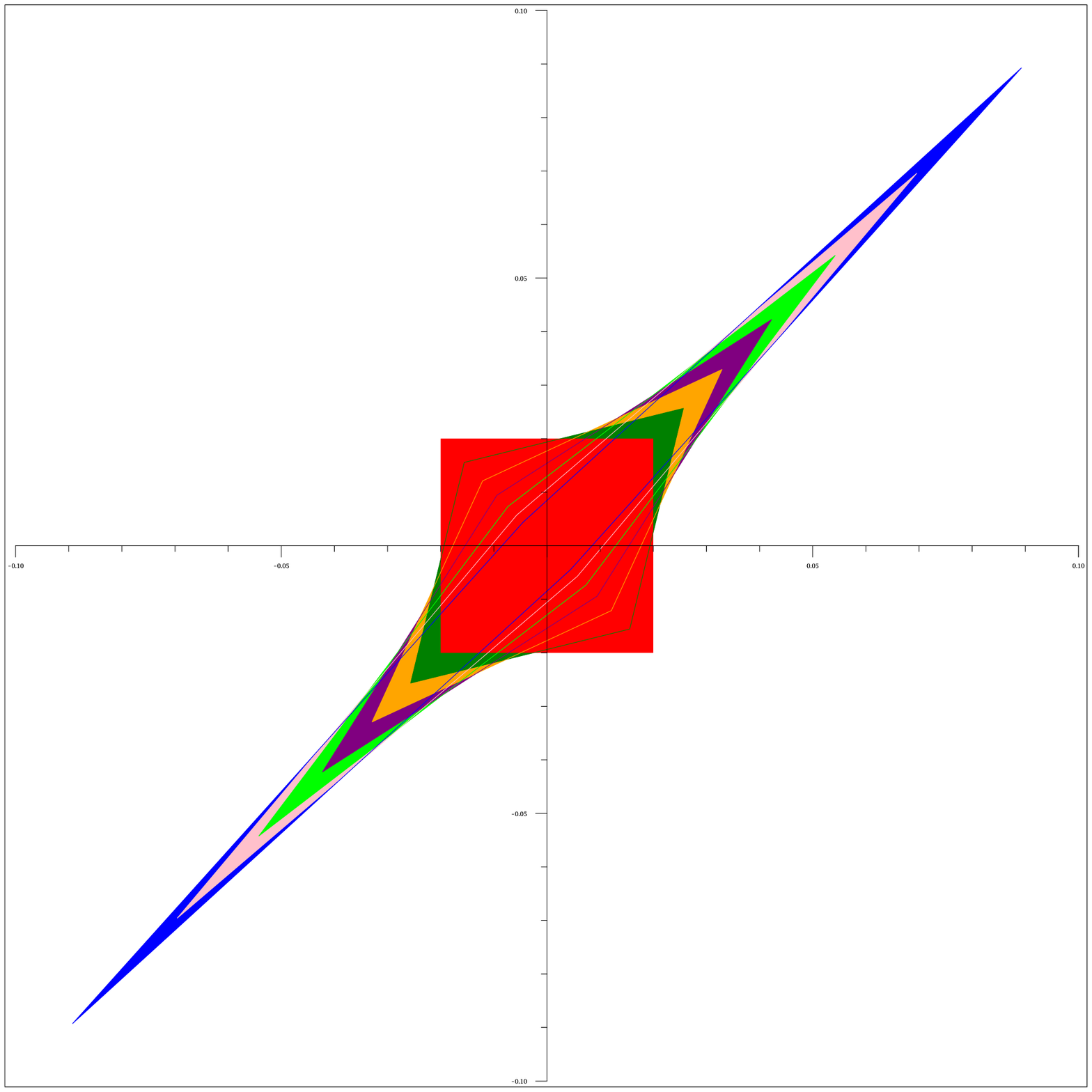}~
\includegraphics[scale=0.385]{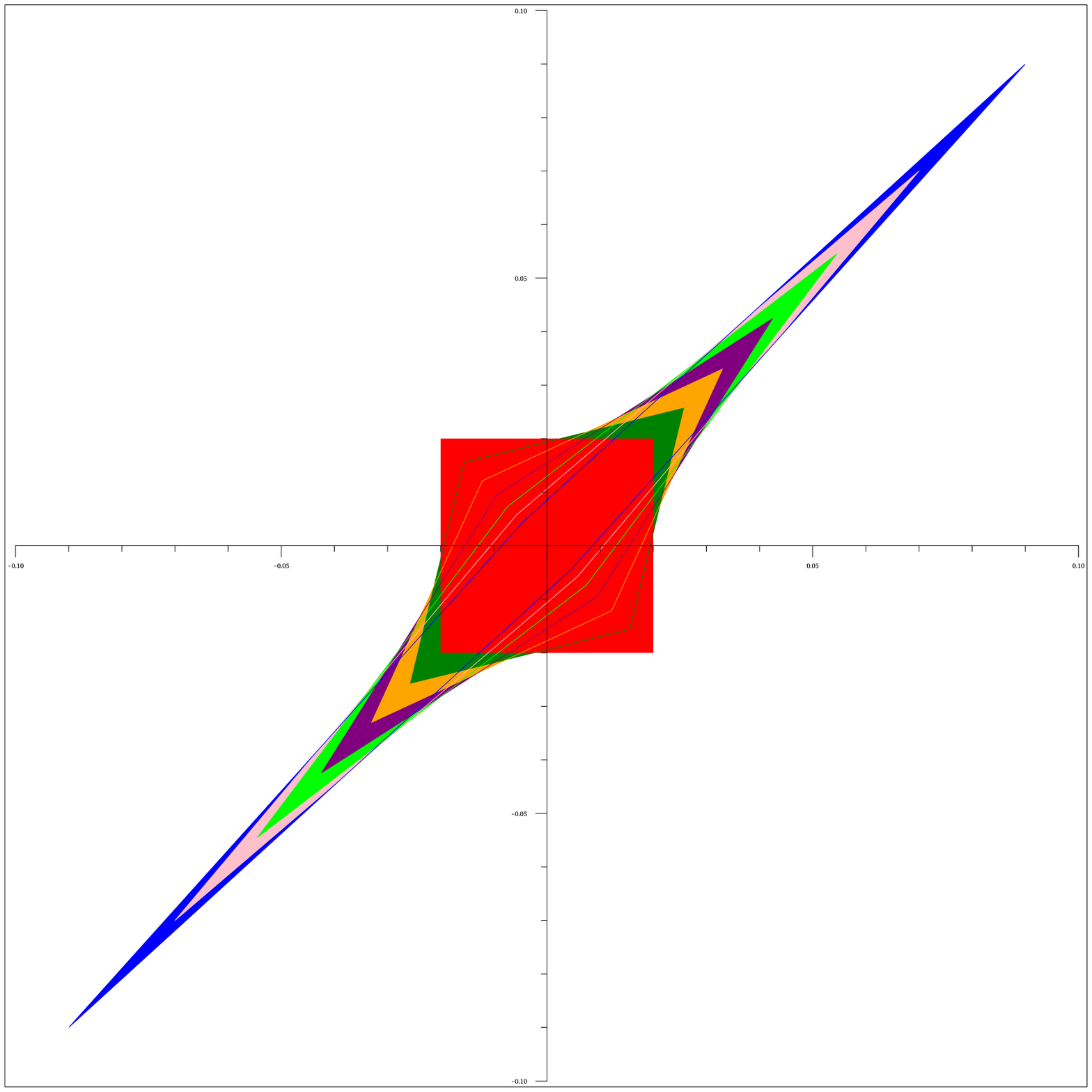}.
\end{center}
\caption{Evolution of MWU (left) and OMWU (right) in zero-sum (top) and coordination (bottom) games in the dual space of Eshel and Akin.
The initial set is the red square. The top two figures were already shown and discussed in the first page.
The bottom two figures correspond to MWU and OMWU in the coordination game $(\bbA,\bbA)$, where $\bbA$ is the $2\times 2$ identity matrix.
The vector fields associated with MWU and OMWU are very similar, and so does the two figures.
However, when we compute how the areas change, we observe that for MWU, the area is shrinking slowly (from red to blue), while for OMWU, the area is increasing slowly.}\label{fig:coor}
\end{figure}

\newpage

\begin{figure}[htp]
\centering
\includegraphics[height=60mm,width=150mm]{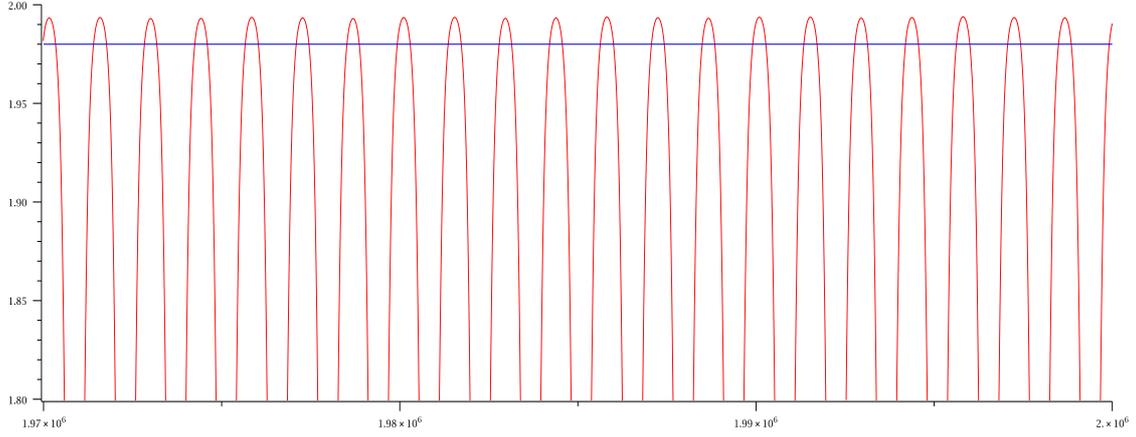}
\caption{Let $\bbx,\bby$ denote respectively the mixed strategies of Players 1 and 2 in the classical Rock-Paper-Scissors game.
We plot the quantity $\sum_{j=1}^3 (x_j)^4 + \sum_{k=1}^3 (y_k)^4$ against time steps between $1.97\times 10^6$ to $2.00\times 10^6$,
when both players employ MWU with step-size $\ep=0.005$, and starting point $\bbx^0\propto (1,1,\exp(1/2))$ and $\bby^0\propto (1,1,\exp(-1/2))$.
%$(\bbx^0,\bby^0)\approx((0.27407,0.27407,0.45186),(0.38365,0.38365,0.23270))$.
When the red curve is above the blue horizontal line, \emph{extremism} occurs, i.e., each player concentrate on one strategy, with some $x_j,y_k\ge 0.995$.
Within the $30000$ time steps, extremism occurs for $22$ periods; each period has length around $350$.}\label{fig:RPS}
\end{figure}

\appendix

\section{Missing Examples and Proof in Section~\ref{sect:prelim}}\label{app:missing-prelim}

In Section~\ref{sect:prelim}, we pointed out two facts: 
\begin{enumerate}
\item[(A)] volume contraction in the dual space does \emph{not} necessarily imply stability in either the dual or the primal space;
\item[(B)] volume expansion in the dual space does \emph{not} necessarily imply instability in the primal space.
\end{enumerate}

To see why (A) is true, consider the following parametrized rectangular set $S(z)$ around the origin in the dual space:
\[
S(z) ~:=~ \{ (\bbp,\bbq)\in \rr^2\times \rr^2 ~\Big|~ |p_1|,|q_1|\le 1/z,~|p_2|,|q_2|\le \sqrt{z} \}.
\]
As $z$ increases, the volume of $S(z) = 1/z$ decreases, but its diameter and the quantities $\max \{p_2-p_1\}$, $\max \{q_2-q_1\}$ are $\Theta(\sqrt{z})$ which increase with $z$.
Also, since $S$ contains the points 
\[
((0,\sqrt{z}),(0,\sqrt{z})), ((0,-\sqrt{z}),(0,-\sqrt{z})),
\]
when the set $S(z)$ is converted to the primal space,
$\tG(S)$ contains points close to 
\[
((0,1),(0,1)),((1,0),(1,0))
\]
as $z\ra \infty$, so the diameter of $\tG(S)$ increases to $2$ as $z\ra \infty$.

To see why (B) is true, consider the following parametrized set $S(z)$ in the dual space:
\[
S(z) ~:=~ \{ (\bbp,\bbq)\in \rr^2\times \rr^2 ~\Big|~ p_2 \ge p_1 + z \text{~and~} q_2\ge q_1 + z, \text{~and~} 0\le p_1,p_2,q_1,q_2\le 3z \}.
\]
It is not hard to compute its volume $4z^4$ which increases with $z$, but its primal counterpart contracts and converges to a single point $((0,1),(0,1))$.

We also note that (B) remains true in the dual space used by Eshel and Akin. An example is 
\[
S(z) = \{((p_1-p_3,p_2-p_3),(q_1-q_3,q_2-q_3)) ~\big|~ z\le p_1-p_3,q_1-q_3\le 2z ~\text{and}~ -2z\le p_2-p_3,q_2-q_3\le -z \}.
\]
The volume of $S(z)$ is $z^4$ which increases with $z$, but its primal counterpart converges to the primal point $((1,0,0),(1,0,0))$ as $z\ra\infty$.

\medskip

Proposition~\ref{prop:dual-expand-to-primal-instability-simpler} follows directly from the following proposition.

\begin{proposition}\label{prop:dual-expand-to-primal-instability}
Let $S$ be a set in the dual space with Lebesgue volume $v$. Also, suppose there exists $j\in S_1$ and $k\in S_2$ such that $\max_{(\bbp,\bbq)\in S} p_j - \min_{(\bbp,\bbq)\in S} p_j \le R_j$
and $\max_{(\bbp,\bbq)\in S} q_k - \min_{(\bbp,\bbq)\in S} q_k \le R_k$. Then for $\beta := \exp\left(\left( \frac{v}{R_j R_k} \right)^{1/(n+m-2)}\right)$, at least one of the followings holds:
\begin{itemize}
\item There exists $j'\in S_1$ such that $\left(\max_{(\bbp,\bbq)\in S} \frac{x_{j'}(\bbp)}{x_j(\bbp)}\right) \Big/ \left(\min_{(\bbp,\bbq)\in S} \frac{x_{j'}(\bbp)}{x_j(\bbp)}\right) \ge \beta$.
Furthermore, if there exists $(\bbp^\#,\bbq^\#) \in S$ such that $x_j(\bbp^\#),x_{j'}(\bbp^\#) \ge \kappa > 0$,
then the diameter of $\tG(S)$ w.r.t.~$\ell_2$ norm is at least $\left( 1 - \beta^{-1/4} \right) \kappa$.
\item There exists $k'\in S_2$ such that $\left(\max_{(\bbp,\bbq)\in S} \frac{y_{k'}(\bbq)}{y_k(\bbq)}\right) \Big/ \left(\min_{(\bbp,\bbq)\in S} \frac{y_{k'}(\bbq)}{y_{k}(\bbq)}\right) \ge \beta$.
Furthermore, if there exists $(\bbp^\#,\bbq^\#) \in S$ such that $y_k(\bbq^\#),y_{k'}(\bbq^\#) \ge \kappa > 0$,
then the diameter of $\tG(S)$ w.r.t.~$\ell_2$ norm is at least $\left( 1 - \beta^{-1/4} \right) \kappa$.
\end{itemize}
\end{proposition}

\begin{pf}%{Proposition~\ref{prop:dual-expand-to-primal-instability}}
Without loss of generality, we assume that $j=1$ and $k=1$.
Consider the mapping:
\[
((p_1,p_2,\cdots,p_n) ~,~ (q_1,q_2,\cdots,q_m)) ~~\rightarrow~~ ((p_1,p_2-p_1,\cdots,p_n-p_1) ~,~ (q_1,q_2-q_1,\cdots,q_m-q_1)).
\]
This is a linear mapping, and it is easy to verify that the determinant of the matrix that describes this linear mapping has determinant $1$, so the mapping is volume-preserving.

Suppose that each of the quantities $p_{j'}-p_1$ and $q_{k'}-q_1$ is bounded by an interval of length at most $R$ within the set $S$, for a value of $R$ to be specified later.
Then $S$ is a subset of a rectangular box in $\rr^{n+m}$, with $n+m-2$ sides of lengths at most $R$, and the remaining two sides of lengths at most $R_j$ and $R_k$.
Thus, the volume of $S$ after the above linear mapping is at most $R^{n+m-2} R_j R_k$. When $R < ( \frac{v}{R_j R_k} )^{1/(n+m-2)}$, this is a contradiction.

Thus, there exists one quantity $p_{j'}-p_1$ or $q_{k'}-q_1$ which is \emph{not} bounded by an interval of length at most $( \frac{v}{R_j R_k} )^{1/(n+m-2)}$.
Then we are done by recalling that $\frac{x_{j'}(\bbp)}{x_1(\bbp)} = \exp(p_{j'}-p_1)$ and $\frac{y_{k'}(\bbq)}{y_1(\bbq)} = \exp(q_{k'}-q_1)$.

If furthermore, there exists $(\bbp^\#,\bbq^\#) \in S$ such that $x_j(\bbp^\#),x_{j'}(\bbp^\#) \ge \kappa > 0$, then
there exists $(\bbp^*,\bbq^*) \in S$ such that either 
\[
\frac{x_j(\bbp^*)}{x_{j'}(\bbp^*)} \left/ \frac{x_j(\bbp^\#)}{x_{j'}(\bbp^\#)} \right. \ge \beta^{1/2}~~~~~~\text{or}~~~~~~
\frac{x_j(\bbp^*)}{x_{j'}(\bbp^*)} \left/ \frac{x_j(\bbp^\#)}{x_{j'}(\bbp^\#)} \right. \le \beta^{-1/2}.
\]
We focus on the former case, as the latter case is similar. We have
$x_j(\bbp^*) - x_j(\bbp^\#) ~\ge~ x_j(\bbp^\#) \cdot \left( \frac{x_{j'}(\bbp^*)}{x_{j'}(\bbp^\#)} \cdot \beta^{1/2}-1 \right)$.
If $\frac{x_{j'}(\bbp^*)}{x_{j'}(\bbp^\#)} \ge \beta^{-1/4}$, we have $x_j(\bbp^*) - x_j(\bbp^\#) \ge \kappa (\beta^{1/4}-1) \ge \kappa (1-\beta^{-1/4})$.
Otherwise, $\frac{x_{j'}(\bbp^*)}{x_{j'}(\bbp^\#)} < \beta^{-1/4}$, and hence $x_{j'}(\bbp^\#) - x_{j'}(\bbp^*) > x_{j'}(\bbp^\#) \cdot \left( 1 - \beta^{-1/4} \right) \ge \kappa (1-\beta^{-1/4})$.
\end{pf}

\section{Extremism of MWU in Zero-Sum Games}\label{app:extremism}

\begin{lemma}\label{lem:single-minded}
Suppose an agent has $m$ options which she use MWU with step-size $\ep$ to decide the mixed strategy $\bby^t = (y_1^t,\cdots,y_m^t)$ in each time step.
Suppose at each round $t$, the payoff to each option $k$ is $a_k + \delta_k^t$, where 
\begin{itemize}
\item each $a_k\in [-1,1]$;
\item there exists a positive number $\alpha_2>0$, such that for any $2\le k \le m$, $a_{k-1} - a_{k} \ge \alpha_2$;
\item there exists a positive number $\delta \le \alpha_2/8$, such that $\delta_k^t\in [-2\delta,2\delta]$.
\end{itemize}
Let $\hk(t)$ denote the strategy $\min \{k\in [m] ~|~ y_k^t > \delta/(m-1)\}$.
%Then $\hk(t)$ is a decreasing function of $t$. Furthermore,
Then for $T:=\ceil{\frac{2}{\ep (\alpha_2 - 4\delta)}\cdot \ln \frac{m-1}{\delta}}$,
(i) if $\bby^{\tau+T}$ has more than one entries larger than $\delta/(m-1)$ for some $\tau\ge 0$, then $\hk(\tau+T) \le \hk(\tau)-1$, and 
(ii) for some $t\le (m-1)T$, $\bby^t$ has an entry which is larger than or equal to $1-\delta$.
\end{lemma}

\begin{proof}
For part (i), we prove the contrapositive statement instead: if $\hk(\tau+T) \ge \hk(\tau)$, then $\bby^{\tau+T}$ has exactly one entry larger than $\delta/(m-1)$.

Let $k = \hk(\tau)$. For any $\ell > k$, due to the definition of the MWU update rule~\eqref{eq:MWU-primal} and our assumptions, for $t\ge \tau$,
\[
\frac{y_\ell^{t+1}}{y_k^{t+1}} ~=~ \frac{y_\ell^{t}}{y_k^{t}} \cdot \exp \left( \ep (a_\ell + \delta_\ell^t - a_k - \delta_k^t) \right)
~\le~ \frac{y_\ell^{t}}{y_k^{t}} \cdot \exp \left( -\ep (\alpha_2 - 4\delta) \right).
\]
Since $k = \hk(\tau)$, we have $y_k^\tau > \delta/(m-1)$. Also, $y_k^t,y_\ell^\tau\le 1$ trivially. Thus, for any $t\ge \tau$,
\[
y_\ell^{t} ~\le~ y_k^{t} \cdot \frac{y_\ell^{\tau}}{y_k^{\tau}} \cdot \exp \left( -\ep (\alpha_2 - 4\delta) (t-\tau) \right)
~<~ \frac{m-1}{\delta} \cdot \exp \left( -\ep (\alpha_2 - 4\delta) (t-\tau) \right).
\]
When $\exp \left( -\ep (\alpha_2 - 4\delta) (t-\tau) \right)\le \delta^2 / (m-1)^2$,
or equivalently $t \ge \tau + \ceil{\frac{2}{\ep (\alpha_2 - 4\delta)}\cdot \ln \frac{m-1}{\delta}} = \tau + T$,
we have $y_\ell^t \le \delta/(m-1)$.

Due to the conclusion of the last paragraph, we have $\hk(\tau+T) \le k$. But we also have the assumption $\hk(\tau+T) \ge \hk(\tau) = k$.
Thus, $\hk(\tau+T) = k$, and hence for any $k'<k$, $y_{k'}^{\tau+T} \le \delta/(m-1)$.
This, together with the conclusion of the last paragraph, shows that $y_k^{\tau+T}$ is the only entry in $\bby^{\tau+T}$ which is larger than $\delta/(m-1)$.
This completes the proof of part (i).

We prove part (ii) by contradiction. Suppose that for all $t\le (m-1)T$, $\bby^t$ has more than one entries larger than $\delta/(m-1)$.
First of all, $\hk(0)\neq m$, for otherwise $y_m^0$ is the only entry in $\bby^0$ which is larger than $\delta/(m-1)$.
Next, we apply part (i) for $(m-1)$ times to yield that $\hk((m-1)T) \le \hk(0) - (m-1) \le 0$, a contradiction.
Thus, for some $\bby^t$ with $t\le (m-1)T$, it has exactly one entry which is larger than $\delta/(m-1)$.
The entry has to be larger than or equal to $1-(m-1)(\delta/(m-1)) = 1-\delta$.
\end{proof}

\begin{pfof}{Theorem~\ref{thm:extremism}}
The proof comprises of three steps.

\smallskip

\parabold{Step 1.}~We show that for any $\kappa > 0$, $\calE^\kappa_{2,2}$ is an uncontrollable primal set
with $\inf_{(\bbx,\bby)\in \calE^\kappa_{2,2}} C(\bbx,\bby) \ge \kappa^2 (\alpha_1)^2 / 2$.

Recall Lemma~\ref{lem:Cxy} that $C(\bbx,\bby)$ is the variance of a random variable $X$, which is equal to $\expect{(X-\expect{X})^2}$.
For any point $(\bbx,\bby)\in \calE^\kappa_{2,2}$, each of $\bbx,\bby$ has at least two entries larger than $\kappa$. Suppose $x_{j_1},x_{j_2},y_{k_1},y_{k_2} > \kappa$.
Then
\begin{equation}\label{eq:variance}
C(\bbx,\bby) ~\ge~ \sum_{j\in \{j_1,j_2\}} ~\sum_{k\in \{k_1,k_2\}}~ \kappa^2 \left[ \underbrace{\left(A_{jk} - [\bbA\bby]_j - [\bbA\trans\bbx]_k\right) - \expect{X}}_{A'_{jk}} \right]^2.
\end{equation}
Due to Condition (A) and Equation~\eqref{eq:distance-from-trivial}, we are guaranteed that
among the four possible values of $A'_{jk}$, the maximum and minimum values differ by at least $\alpha_1$,
for otherwise we can choose $a_j = [\bbA\bby]_j + \expect{X}$ and $b_k = -[\bbA\trans\bbx]_k$ in \eqref{eq:distance-from-trivial}
to show that the $2\times 2$ sub-matrix of $\bbA$ corresponding to strategies $\{j_1,j_2\}\times \{k_1,k_2\}$ has distance from triviality strictly less than $\alpha_1$.
Consequently, $C(\bbx,\bby) \ge \kappa^2 (\alpha_1/2)^2 \cdot 2 = \kappa^2 (\alpha_1)^2/2$.

\smallskip

\parabold{Step 2.}~Then we apply Theorem~\ref{thm:unavoidable} to show that for any step-size $\ep < \min \left\{\frac{1}{32n^2 m^2}~,~\frac{\kappa^2 (\alpha_1)^2}{2} \right\}$,
there exists a dense subset of points in $\interior(\Delta)$ such that the flow of each such point must eventually reach a point outside $\calE^\kappa_{2,2}$.
Let $(\hbbx,\hbby)$ denote the point outside $\calE^\kappa_{2,2}$.
At $(\hbbx,\hbby)$, one of the two players, which we assume to be Player 1 without loss of generality,
concentrates her game-play on only one strategy, which we denote by strategy $\hj$.
Precisely, for any $j\neq \hj$, $\hat{x}_j \le \kappa$, and hence $\sum_{j\in S_1\setminus\{\hj\}} \hat{x}_j \le (N-1)\kappa$.

\smallskip

\parabold{Step 3.}~Now, we consider the flow starting from $(\hbbx,\hbby)$. Since $x_j^{t+1} / x_j^t \le \exp(2\ep)$ always,
we are sure that for the next $T_1 := \floor{\frac{1}{2\ep}\ln \frac{\delta}{(N-1)\kappa}}$ time steps, $\sum_{j\in S_1\setminus\{\hj\}} x_j^t \le \delta$.
Thus, within this time period, the payoff to strategy $k$ of Player 2 in each time step is $-A_{\hj k}$ plus a perturbation term in the interval $[-2\delta,2\delta]$.
Then by Lemma~\ref{lem:single-minded} part (ii) (a sanity check on the conditions required by the lemma is easy and thus skipped),
if $(N-1)\cdot \ceil{\frac{2}{\ep (\alpha_2-4\delta)}\cdot \ln \frac{N-1}{\delta}} \le T_1$, we are done.
A direct arithmetic shows that this inequality holds if $\kappa \le (\delta/(N-1))^{4(N-1)/(\alpha_2-4\delta)+1}/3$.
\end{pfof}

\begin{pfof}{Theorem~\ref{thm:extremal-io}}
By Theorem~\ref{thm:extremism}, we are guaranteed that there exists a dense set of starting points such that
the flow of each of them must eventually reach the extremal domain with threshold $\delta$.
When we apply Theorem~\ref{thm:extremism}, 
This is our starting point to prove Theorem~\ref{thm:extremal-io}.

\smallskip

\parabold{Step 1.} We show that: for each such starting point $y$, we prove that its flow cannot remain in the extremal domain forever.

First, observe that the extremal domain is the union of small neighbourhoods of extremal points, and
each such neighbourhood is far from the other neighbourhoods.

Suppose the contrary that there exists a starting point such that its flow remains in the extremal domain forever.
Due to the above observation, its flow must remain in the small neighbourhood of \emph{one} extremal point forever.
Suppose the utility values at this extremal point is $(u,-u)$; recall that by assumption, $|u-v| \ge r$.
Since the flow remains near this extremal point, in the long run, the average utility gained by Player 1
must lie in the interval $(1-\delta)u \pm \delta$, which is a subset of the interval $u \pm 2\delta$.

On the other hand, due to a well-known regret bound of MWU (see, for instance,~\cite[Lemma 9]{Cheung2018}),
in the long run, the average utility gained by Player 1 must lie in the interval $v\pm 3\ep$.
When $3\ep + 2\delta \le r/2$, this is incompatible with the interval derived in the previous paragraph,
thus a contradiction.

\smallskip

\parabold{Step 2.} Indeed, we have a stronger version of the result in Step 1. Recall that the complement of the extremal domain is an open set.
Since the MWU update rule is a continuous mapping, it preserves openness,
and hence we not only one point $y$ that visits and leaves the extremal domain,
but we have an open neighbourhood $\calO_1$ around $y$, such that the flow of $\calO_1$ visits and leaves the extremal domain.
Let $\calO'$ denote the flow of $\calO_1$ at the moment when the flow leaves the extremal domain. $\calO'$ is open, and hence has positive Lebesgue measure.

Then we construct a closed subset $\calC_1 \subset \calO_1$ with positive Lebesgue measure. This is easy as follows.
First, we take an arbitrary point $z \in \calO'$. Since $\calO'$ is open, there exists an open ball around $z$ with some radius $r>0$ which is contained in $\calO'$.
Since the MWU update rule is a continuous mapping, its inverse for arbitrary finite time preserves closeness,
the inverse (back to the starting time) of the closed ball around $z$ with radius $r/2$ is a closed set, which we take as $\calC_1$;
$\calC_1\subset \calO_1$ since the closed ball around $z$ with radius $r/2$ is a subset of $\calO'$, and the inverse (back to the starting time) of $\calO'$ is $\calO_1$.

\smallskip

\parabold{Step 3.} Since $\calC_1$ has positive Lebesgue measure,
we can reiterate the arguments in Steps 1 and 2, and construct open set $\calO_2\subset \calC_1$ and closed set $\calC_2\subset \calO_2$ that visit and leave the extremal domain again.

By iterating these arguments repeatedly, we get a sequence of closed (and indeed compact) sets $\calC_1 \supset \calC_2 \supset \calC_3 \supset \cdots$.
By the Cantor's intersection theorem, the intersection of this sequence of closed sets must be non-empty.
Then any point in this intersection is a starting point that visits and leaves the extremal domain infinitely often.
\end{pfof}

\subsection{Classical Rock-Paper-Scissors Game}\label{app:RPS}

The standard Rock-Paper-Scissors game is the zero-sum game $(\bbA,-\bbA)$ with the following payoff matrix:
$
\bbA =\left[\begin{smallmatrix}
0 & -1 & 1\\
1 & 0 & -1\\
-1 & 1 & 0
\end{smallmatrix}\right]
$.
There are two types of $2\times 2$ sub-matrices of $\bbA$. Consider such a sub-matrix which corresponds to strategy set $Q_i\subset \{R,P,S\}$ for Players $i=1,2$.
The first type is when $Q_1 = Q_2$, then the sub-matrix is $\bbA' = \left[\begin{smallmatrix} 0 & -1\\ 1 & 0 \end{smallmatrix}\right]$, which is trivial, i.e., $c(\bbA') = 0$.
The second type is when $|Q_1\cap Q_2| = 1$, then the sub-matrix is $\bbA'' = \left[\begin{smallmatrix} 0 & 1\\ 1 & -1 \end{smallmatrix}\right]$; it is easy to show that $c(\bbA'') = 3/2$.
Due to the existence of the first type of sub-matrices, Theorem~\ref{thm:extremism} cannot be applied.
We provide a separate proof to show that the same conclusion of Theorem~\ref{thm:extremism} holds for this specific game.

\begin{theorem}\label{thm:RPS}
Suppose the underlying game is the standard Rock-Paper-Scissors game. For any $0<\delta < 1/20$,
if both players use MWU with step-size $\ep$ satisfying $\ep < \delta^{22}/(34\times 10^6)$,
then there exists a dense subset of points in $\interior(\Delta)$, such that the flow of each such point must eventually reach a point $(\bbx,\bby)$
where each of  $\bbx,\bby$ has exactly one entry larger than or equal to $1-\delta$.
\end{theorem}

\begin{proof}
To start, we define a new family of primal set $\calE^\kappa$. To define it, let $(\bbx,\bby)$ be a point in $\interior(\Delta)$,
and let $Q_i$ denote the set of strategies of Player 1 with probability density larger than $\kappa$.
Then $(\bbx,\bby) \in \calE^\kappa$ if and only if $|Q_1|,|Q_2|\ge 2$, and furthermore, there exists $Q'_1\subset Q_1$, $Q'_2\subset Q_2$ such that
$|Q'_1|,|Q'_2|=2$ and $|Q'_1\cap Q'_2| = 1$.

The definition of $\calE^\kappa$ deliberately avoids us from deriving a lower bound of $C(\bbx,\bby)$ in the manner of~\eqref{eq:variance} when $\{j_1,j_2\} = \{k_1,k_2\}$,
which corresponds to a trivial sub-matrix.
Then by following Step 1 in the proof of Theorem~\ref{thm:extremism}, we have $\inf_{\bbx,\bby\in \calE^\kappa} \ge \kappa^2 c(\bbA'')^2/2 = 9\kappa^2 / 8$.
By following Step 2 in the proof of Theorem~\ref{thm:extremism}, when $\ep < \min \left\{ \frac{1}{32n^2 m^2}~,~\frac{9\kappa^2}{8} \right\}$,
there exists a dense set of points in $\interior(\Delta)$ such that the flow of each such point must reach a point $(\hbbx,\hbby)$ outside $\calE^\kappa$.

% when the flow starts from a point $(\hbbx,\hbby)$ which is outside $\calE^\kappa$.
Below, we assume the time is reset to zero with starting point $(\hbbx,\hbby)$.
We proceed on a case analysis below.

\medskip

\parabold{Case 1: either $|Q_1|=1$ or $|Q_2|=1$.} For this case, we can simply follow Step 3 in the proof of Theorem~\ref{thm:extremism}.
$\kappa \le \delta^{11}/6144$ suffices.

\medskip

\parabold{Case 2: $Q_1=Q_2$, and $|Q_1|=2$.} Without loss of generality, we assume $Q_1=Q_2=\{R,P\}$.
In the sub-game corresponding to $Q_1\times Q_2$, each player has a strictly dominant strategy, namely $P$.
Intuitively, the probability of choosing strategy $P$ must strictly increase with time (when we ignore the tiny effect of strategy $S$).

More formally, starting from time zero, for the next $T_1 := \floor{\frac{1}{2\ep}\ln \frac{\delta}{2\kappa}}$ time steps, $x_S^t,y_S^t \le \delta/2$,
and hence $x_P^t+x_R^t,y_P^t+y_R^t \ge 1-\delta/2$. Then
\begin{align*}
&{\small(\text{the payoff to strategy }P\text{ of Player 1 in round }t) - (\text{the payoff to strategy }R\text{ of Player 1 in round }t)}\\
=~&\left[y_P^t \cdot 0 + y_R^t \cdot 1 + y_S^t \cdot (-1)\right] - \left[y_P^t \cdot (-1) + y_R^t \cdot 0 + y_S^t \cdot 1\right]\\
\ge~&y_P^t + y_R^t -\delta ~\ge~ 1-2\delta.
\end{align*}
Thus, $\frac{x_P^{t+1}}{x_R^{t+1}} ~\ge~ \frac{x_P^t}{x_R^t}\cdot \exp \left( \ep (1-2\delta) \right)$, and hence
\begin{equation}\label{eq:MWU-elem-prop}
%x_R^t ~\le~ \frac{1}{\hat{x}_P} \cdot \exp \left( -\ep (1-2\delta)t \right)~~~~~~\text{and}~~~~~~
\frac{x_P^t}{x_R^t} ~\ge~ \hat{x}_P \cdot \exp \left( \ep (1-2\delta)t \right).
\end{equation}
The above inequality holds also when all $x$'s are replaced by $y$'s.

\begin{itemize}%[leftmargin=0.2in]
\item \bold{Case 2(a): at $(\hbbx,\hbby)$, each of the two players have one strategy with probability larger than or equal to $1-\delta$.} Then we are done.

\item \bold{Case 2(b): at $(\hbbx,\hbby)$, each of the two players have all strategies with probability less than $1-\delta$.}
Then we know that $\hat{x}_P,\hat{y}_P\ge 1-(1-\delta)-\delta/2 = \delta/2$.
By~\eqref{eq:MWU-elem-prop}, when $\exp \left( \ep (1-2\delta)t \right)\ge 4/\delta^2$, we have $x^t_P/x^t_R , y^t_P/y^t_R \ge 2/\delta$.
And since we still have $x^t_S,y^t_S\le \delta/2$, it is easy to show that $x^t_P,y^t_P \ge 1-\delta$.

\item \bold{Case 2(c): at $(\hbbx,\hbby)$, exactly one of the two players have one strategy with probability larger than or equal to $1-\delta$.}
Without loss of generality, we assume the player is Player 2. Then we know that $\hat{x}_P,\hat{x}_R\ge \delta/2$.
Similar to the argument for Case 2(b), when $\exp \left( \ep (1-2\delta)t \right)\ge 4/\delta^2$, we have $x^t_P \ge 1-\delta$.

If at this time $t$, we have either $y^t_P\ge 1-\delta$ or $y^t_R\ge 1-\delta$, we are done.
Otherwise, we have $y^t_P\ge \delta/2$. Thus, after another period of time $t'$ such that $\exp \left( \ep (1-2\delta)t' \right)\ge 4/\delta^2$,
we have $y^{t+t'}_P\ge 1-\delta$, while $x^{t+t'}_P\ge 1-\delta$ still.
\end{itemize}

For the arguments for Cases 2(b),(c) to hold, we need
\[
2\cdot \ceil{\frac{1}{(1-2\delta)\ep}\ln \frac{4}{\delta^2}} ~~\le~~ T_1,
\]
A direct arithmetic shows that $\kappa \le \delta^{10}/2845$ suffices.
\end{proof}
\section{Continuous Analogue of OMWU in General-Sum Games}\label{app:cont-analogue-GS-games}

In equations~\eqref{eq:recur-p} and~\eqref{eq:recur-q}, 
observe that each $\difft{p_j}$ is expressed as an affine combination of various $\difft{q_k}$,
while each $\difft{q_k}$ is expressed as an affine combination of various $\difft{p_j}$.
Thus, we may rewrite all these expressions into a matrix-form differential equation. Let $\bbv(\bbp,\bbq)$ denote the following vector in $\rr^{n+m}$:
\[
\bbv(\bbp,\bbq) = ([\bbA \cdot \bby(\bbq)]_1,\cdots,[\bbA \cdot \bby(\bbq)]_n~,~[\bbB\trans \cdot \bbx(\bbp)]_1,\cdots,[\bbB\trans \cdot \bbx(\bbp)]_m)\trans~,
\]
and let $\bbM(\bbp,\bbq)$ denote the $(S_1\cup S_2)\times (S_1\cup S_2)$ matrix $\left[\begin{smallmatrix} \mathbf{0} & \bbM^1\\ \bbM^2 & \mathbf{0} \end{smallmatrix}\right]$,
where $\bbM^1\equiv \bbM^1(\bbp,\bbq)$ is a $S_1\times S_2$ sub-matrix and $\bbM^2 \equiv \bbM^2(\bbp,\bbq)$ is a $S_2\times S_1$ sub-matrix defined as below:
\[
M^1_{jk} = y_k(\bbq) \cdot (A_{jk} - [\bbA \cdot \bby(\bbq)]_j)~~~~\text{and}~~~~M^2_{kj} = x_j(\bbp) \cdot (B_{jk} - [\bbB\trans \cdot \bbx(\bbp)]_k).
\]
Then we can rewrite the recurrence system~\eqref{eq:recur-p} and~\eqref{eq:recur-q} as
$\left( \difft{\bbp} ~,~ \difft{\bbq} \right)\trans = \bbv(\bbp,\bbq) + \ep \cdot \bbM(\bbp,\bbq) \cdot \left( \difft{\bbp} ~,~ \difft{\bbq} \right)\trans$.
This can be easily solved to a standard (non-recurring) system of ODE:
\[
\left( \difft{\bbp} ~,~ \difft{\bbq} \right)\trans ~=~ \left( \bbI - \ep \cdot \bbM(\bbp,\bbq) \right)^{-1} \cdot \bbv(\bbp,\bbq),
\]
if the inverse of the matrix $(\bbI - \ep \cdot \bbM(\bbp,\bbq))$ exists.

We proceed by using the following identity: if a square matrix $\bbR$ satisfies $\sup_{\|\bbz\|=1} \|\bbR \bbz\| < 1$, then $(\bbI - \bbR)^{-1} = \bbI + \sum_{\ell=1}^{\infty} \bbR^\ell$.
In our case, we desire $\sup_{\|\bbz\|=1} \|\ep \cdot \bbM(\bbp,\bbq) \cdot \bbz\| ~<~ 1$.
Observe that for each row of $\bbM(\bbp,\bbq)$, its $\ell_2$-norm is at most $2\|\bbx\|$ or $2\|\bby\|$, which are upper bounded by $2$.
Thus, each entry in $\ep \cdot \bbM(\bbp,\bbq) \cdot \bbz$ is absolutely bounded by $2\ep$, and hence $\|\ep \cdot \bbM(\bbp,\bbq) \cdot \bbz\| \le 2\ep \sqrt{n+m}$.
Consequently, $\ep < 1/(2\sqrt{n+m})$ suffices to guarantee that the inverse of $(\bbI - \ep \cdot \bbM(\bbp,\bbq))$ exists, and the identity mentioned above holds for its inverse:
\[
\left( \difft{\bbp} ~,~ \difft{\bbq} \right)\trans ~=~ \left( \bbI + \sum_{\ell=1}^\infty \ep^\ell \cdot \bbM(\bbp,\bbq)^\ell \right) \cdot \bbv(\bbp,\bbq).
\]

\section{Volume Analysis of Discrete-Time OMWU}\label{app:vol-analysis-OptMD}

Recall from~\cite{CP2019} that the volume integrand for MWU is
\[
1 ~+~ C_{(\bbA,\bbB)}(\bbp,\bbq) \cdot \ep^2 ~+~ \calO(\ep^4),
\]
while by~\eqref{eq:integrand-OptMD}, the volume integrand for OMWU is
\[
1 ~-~ C_{(\bbA,\bbB)}(\bbp,\bbq) \cdot \ep^2 ~+~ \calO(\ep^3).
\]
By~\eqref{eq:calE-uncontrol}, within $\tGinv(\calE^\delta_{2,2})$, $C_{(\bbA,-\bbA)}(\bbp,\bbq) \ge \delta^2 (\alpha_1)^2 / 2$,
thus $C_{(\bbA,\bbA)}(\bbp,\bbq) ~=~ -C_{(\bbA,-\bbA)}(\bbp,\bbq)$ $\le -\delta^2 (\alpha_1)^2 / 2$.
Therefore, when $\ep$ is sufficiently small, the volume integrands for MWU in coordination game and OMWU in zero-sum game are both at most
$1- \ep^2 \delta^2 (\alpha_1)^2 / 4$.

\begin{theorem}\label{thm:OMD-volume-contract-zerosum-discrete}
Suppose the underlying game is a non-trivial zero-sum game $(\bbA,-\bbA)$ and the parameter $\alpha_1$ as defined in Theorem~\ref{thm:extremism} is strictly positive.
For any $1/2 > \delta > 0$, for any sufficiently small $0<\ep\le \bar{\ep}$ where the upper bound depends on $\delta$,
and for any set $S=S(0)\subset \tGinv(\calEd_{2,2})$ in the dual space, 
if $S$ is evolved by the OMWU update rule~\eqref{eq:OptMD} and if its flow remains a subset of $\tGinv(\calEd_{2,2})$ for all $t\le T-1$, then
$\vol(S(T)) ~\le~ \left( 1 - \frac{\ep^2 \delta^2 (\alpha_1)^2}{4} \right)^T \cdot \vol(S).$
\end{theorem}

\begin{theorem}\label{thm:MWU-volume-contract-coordination-discrete}
Suppose the underlying game is a non-trivial coordination game $(\bbA,\bbA)$ and the parameter $\alpha_1$ as defined in Theorem~\ref{thm:extremism} is strictly positive.
For any $1/2 > \delta > 0$, for any sufficiently small $0<\ep\le \bar{\ep}$ where the upper bound depends on $\delta$,
and for any set $S=S(0)\subset \tGinv(\calEd_{2,2})$ in the dual space, 
if $S$ is evolved by the MWU update rule~\eqref{eq:MWU} and if its flow remains a subset of $\tGinv(\calEd_{2,2})$ for all $t\le T-1$, then
$\vol(S(T)) ~\le~ \left( 1 - \frac{\ep^2 \delta^2 (\alpha_1)^2}{4} \right)^T \cdot \vol(S).$
\end{theorem}
\end{document}